\documentclass[10pt,a4paper]{article}
\usepackage[margin=3cm]{geometry}

\usepackage{microtype}
\usepackage[draft]{graphicx}
\usepackage{xcolor}
\usepackage[colorlinks=true,hyperindex=true]{hyperref}
\hypersetup{colorlinks,citecolor=violet,linkcolor=violet,urlcolor=violet,pdfnewwindow=true,pdfstartview={FitH}}
\usepackage[shortlabels]{enumitem}
\usepackage{booktabs}
\usepackage{multicol}
\usepackage[normalem]{ulem}

\usepackage{amsmath,amsfonts,amssymb,amsthm}
\numberwithin{equation}{section}
\usepackage{mathtools}
\mathtoolsset{showonlyrefs}

\newtheorem{theorem}{Theorem}[section]
    
    \newtheorem{proposition}[theorem]{Proposition}
    \newtheorem{lemma}[theorem]{Lemma}
    
    \newtheorem{thought}[theorem]{Thought Experiment}

\theoremstyle{definition}
    
    \newtheorem{remark}[theorem]{Remark}
     \newtheorem{example}[theorem]{Example}

\def\rr{{\mathbb R}}
\def\qq{{\mathbb Q}}
\def\zz{{\mathbb Z}}
\def\cc{{\mathbb C}}
\def\nn{{\mathbb N}}
\def\ss{{\mathbb S}}

\def\cA{{\cal A}}
\def\cB{{\cal B}}

\renewcommand{\sp}{\operatorname{sp}}

	\makeatletter
	\providecommand*{\diff}%
	{\@ifnextchar^{\DIfF}{\DIfF^{}}}
	\def\DIfF^#1{%
	\mathop{\mathrm{\mathstrut d}}%
	\nolimits^{#1}\gobblespace}
	\def\gobblespace{%
	\futurelet\diffarg\opspace}
	\def\opspace{%
	\let\DiffSpace\!%
	\ifx\diffarg(%
	\let\DiffSpace\relax
	\else
	\ifx\diffarg%
	\let\DiffSpace\relax
	\else
	\ifx\diffarg\{%
	\let\DiffSpace\relax
	\fi\fi\fi\DiffSpace}

\newcommand{\dd}{\diff}
\DeclareMathOperator{\tr}{tr}

\newcommand{\one}{\mathbf{I}}
\newcommand{\Exp}[1]{\mathrm{e}^{#1}}

\newcommand{\pp}{\mathbb{P}}

\newcommand{\opone}{_{\textnormal{op},1}}
\newcommand{\opty}{_{\textnormal{op},\infty}}
\newcommand{\Zp}{\tilde{Z}}
    
\title{Quenched large deviations of Birkhoff sums \\ along random quantum measurements}
\author{R.\ Raqu\'epas\footnote{New York University, Courant Institute of Mathematical Sciences, New York NY, United States. \\ \phantom{Au\,} Corr.: rr4374@nyu.edu}\: and J.\ Schenker\footnote{Michigan State University, Department of Mathematics, East Lansing MI, United States.}}
\date{}

\begin{document}

\maketitle

\begin{abstract}
    We prove a quenched version of the large deviation principle for Birkhoff-like sums along a sequence of random quantum measurements driven by an ergodic process. We apply the result to the study of entropy production in the two-time measurement framework.\\

    \noindent \emph{Keywords} \quad large deviation principle, quantum instruments, quantum channels, quenched disorder, entropy production \\

    \noindent \emph{MSC2020} \quad
    60F10, % large deviations
    81P15, % quantum measurement theory
    81S22 % open systems, reduced dynamics
    82C10, % quantum dyn and nonequilibrium stat mech
    82C44 % dynamics of disordered systems in time-dep stat mech
    
\end{abstract}

\section{Introduction}

Statistical properties of outcomes for repeated quantum measurements have attracted considerable attention in both the physics and mathematics literature over the past couple of decades; see e.g.~\cite{KM04,BJM06,BB11,vHG15,AGPS15,BJPP18}.\footnote{It was later realized that earlier works such as~\cite{Ku89,FNW92,Fe03}\,---\,and the parallel literature that followed\,---\,contained important insights on the statistical properties of such systems.} Skipping the details of the physical interpretation for the time being, the main object of interest is a measure on a shift space of the form~$\cA^\nn$ with marginals $(p^{(n)})_{n\in\nn}$ of the form
\[ 
    p^{(n)}(a_1, \dotsc, a_{n}) = \tr[(\psi_{n,a_{n}} \circ \dotsb \circ \psi_{1,a_1})(\rho)]
\]
describing the probability that the first $n$ measurements lead to results labeled by $(a_j)_{j=1}^n$ (elements of~$\cA$). Here, $\rho$ is a density matrix and the maps $\psi_{j,a_j}$ are completely positive and satisfy suitable normalization conditions; see Section~\ref{sec:setup}.
Together with these marginals are often considered random variables~$(\Sigma_n)_{n\in\nn}$ of the Birkhoff-sum form
\[
    \Sigma_n(a_1, \dotsc, a_{n}) =  \sum_{j=1}^{n} f_j(a_j) .
\]
Under different technical conditions, the large deviations of the sequence $(\Sigma_n)_{n\in\nn}$ with respect to~$(p^{(n)})_{n\in\nn}$\,---\,and generalizations thereof\,---\,have been studied in the mathematical physics literature in different situations: when the instruments are identical~\cite{BJPP18,CJPS19,BCJP21}, when the instruments follow a deterministic, quasistatic process~\cite{HJPR18}, and when the instruments are sampled from a Bernoulli or Markov process~\cite{BB20,BJP22}. 

The goal of this short paper is to show that recent results on the ergodic theory of quantum channels~\cite{MS22,PS23} allow us to treat situations where the instruments are sampled from an ergodic stochastic process obeying certain non-degeneracy conditions.  A key strength of the approach employed here is the stochastic process for the instruments need not be Markovian and may in fact have long-range correlations in time. In cases involving randomness, it is important to make the distinction between the so-called ``quenched'' and ``annealed'' regimes; our results concern the former, where one seeks fluctuation theorems that hold for almost all realizations of the underlying random process.

\paragraph*{Organization of the paper} After setting the stage, we present our main large-deviation result in Section~\ref{sec:setup}; the proof is postponed to Section~\ref{sec:proof}. We then move to applications to entropy production in Section~\ref{sec:EP}; the proofs are postponed to Section~\ref{sec:EP-proofs}. 
Finally, Appendix~\ref{app:together} gathers technical lemmas in measure theory and ergodic theory that are used in Section~\ref{sec:proof}.

\paragraph*{Acknowledgements} R.R.\ was partially funded by the \emph{Fonds de recherche du Qu\'ebec\,---\,Nature et technologies} (FRQNT) and by the Natural Sciences and Engineering Research Council of Canada (NSERC). J.S.\ was supported in part by the US National Science Foundation, Grant No.~DMS-2153946.

\section{Setup and main result}
\label{sec:setup}

Let $(\Omega,\theta,\pp)$ be an invertible dynamical system, where $\pp$ is a probability measure on $\Omega$ and $\theta:\Omega\to \Omega$ is an invertible map that is measure preserving and ergodic. This ergodic process is thought of as guiding a \emph{quantum measurement process} (QMP) on a quantum system described using the finite-dimensional Hilbert space~$\cc^d$. At each step of this measurement process, the outcomes are labeled using a fixed finite set~$\cA$.

We equip the space~$\cB(\cc^d)$ with the trace norm $\|\,\cdot\,\|_1$. We will use $\mathbb{S}$ (resp.~$\mathbb{S}^{\circ}$) for the convex subset of positive semidefinite (resp.\ positive definite) matrices in $\cB(\cc^d)$ with trace (norm) 1. Elements of~$\mathbb{S}$ are called \emph{density matrices} and describe the internal state of the aforementioned quantum system. We use $\mathcal{C}(\mathbb{S})$ (resp.\ $\mathcal{C}(\mathbb{S}^{\circ})$) for the closed (resp.\ open), convex cone generated by $\mathbb{S}$ (resp.~$\mathbb{S}^{\circ}$), which is nothing but the set of positive semidefinite (resp.\ definite) matrices in $\cB(\cc^d)$. 

A linear map~$\phi: \cB(\cc^d)\to\cB(\cc^d)$ is called \emph{positive} if it maps $\mathcal{C}(\mathbb{S})$ to $\mathcal{C}(\mathbb{S})$. A positive linear map~$\phi: \cB(\cc^d)\to\cB(\cc^d)$ is called \emph{completely positive} if~$\phi \otimes \operatorname{id}_{\cB(\cc^m)} \ : \ \cB(\cc^{d m}) \to \cB(\cc^{d m})$ is positive for all~$m\in\nn$.
A completely positive linear map~$\phi: \cB(\cc^d)\to\cB(\cc^d)$ is called \emph{trace preserving} (CPTP for short) if it maps $\mathbb{S}$ to itself.
A completely positive linear map~$\phi: \cB(\cc^d)\to\cB(\cc^d)$ is called \emph{positivity improving} if it maps $\mathcal{C}(\mathbb{S})\setminus\{0\}$ to~$\mathcal{C}(\mathbb{S}^{\circ})$, \emph{primitive} if $\phi^{j}$ is positivity improving for some~$j < d$, and \emph{irreducible} if $(\operatorname{id}_{\cB(\cc^d)} + \phi)^{d-1}$ is positivity improving. As is the case for their counterparts for stochastic matrices, these notions have a number of equivalent reformulations and we refer the reader to~\cite{EHK78,FF09}. We will often use the following basic fact: if $\phi$ is CPTP, then $\phi^*$ is \emph{unital} in the sense that $\phi^*(\one) = \one$. Here, the adjoint ``$^*$'' is understood with respect to the Hilbert--Schmidt inner product on $\cB(\cc^d)$. We at times think of~$\phi^*$ as a linear operator on the dual space $(\cB(\cc^d),\|\,\cdot\,\|_\infty)$; at others, on the original space $(\cB(\cc^d),\|\,\cdot\,\|_1)$. We keep in mind that the induced operator norms, namely $\|\,\cdot\,\|\opty$ and $\|\,\cdot\,\|\opone$, are equivalent.

We consider an irreducible \emph{instrument-valued} map $\omega  \mapsto (\psi_{\omega,a})_{a\in\cA}$. By this, we mean that to every~$\omega \in \Omega$ is assigned a tuple of completely positive maps on~$\cB(\cc^d)$\,---\,the dimension~$d$ is fixed\,---\,whose sum 
\[ 
    \varphi_\omega := \sum_{a \in \cA} \psi_{\omega,a}
\]
is an irreducible CPTP map. For fixed $\omega$, i.e.\ quenched disorder, we consider the QMP induced by the instruments $(\psi_{\omega,a})_{a\in\cA}$, $(\psi_{\theta\omega,a})_{a\in\cA}$, $(\psi_{\theta^2\omega,a})_{a\in\cA}$, $\dotsc$ and a given initial state $\rho$. Let us emphasize that the requirement that $\varphi_\omega$ be irreducible is enforced throughout even where not mentioned explicitly.

\begin{example}
\label{ex:examples-intro}
    An important basic example to keep in mind is the following.
    \begin{enumerate}
        \item[i.] Consider $\Omega$ to be a subshift of $\mathcal{S}^\zz$ for some finite or countably infinite set $\mathcal{S}$, the dynamics $\theta$ to be the left shift and $\pp$ to be an ergodic measure for $(\Omega,\theta)$\,---\,e.g.\ a Bernoulli measure or the path measure for an irreducible, stationary Markov chain.\footnote{In this setting, and assuming stationarity, irreducibility of the transition matrix used to define the Markov measure is equivalent to ergodicity for~$(\Omega,\pp,\theta)$; irreducibility \emph{and} aperiodicity, to the ($\beta$-)mixing property~\cite{Br05}.} In such situations, the simplest non-trivial type of instrument-valued map would be a choice of $(\psi_{\omega,a})_{a\in\cA}$ that depends on the zeroth symbol in the fixed~$\omega$, so that $(\psi_{\theta\omega,a})_{a\in\cA}$ depends on the first symbol in the fixed sequence, $(\psi_{\theta^2\omega,a})_{a\in\cA}$ depends on the second symbol in the fixed sequence, and so on.
    \end{enumerate}
    The setup also accommodates more ``continuous'' situations,
    \begin{enumerate}
        \item[ii.]  Consider $\Omega$ to be a symplectic manifold,~$\theta$ to be the time-1 map for a Hamiltonian flow, and $\pp$ be a $\theta$-ergodic Borel measure on $\Omega$. Since $\pp$ is ergodic, it is supported on a constant energy surface in~$\Omega$. In this situation the sequence $(\psi_{\omega,a})_{a\in \mathcal{A}} ,(\psi_{\theta\omega,a})_{a\in \mathcal{A}}, \dotsc$ describes a deterministic change in the measuring apparatus between measurements. 
    \end{enumerate}
\end{example}

 In the quantum measurement process, the probability of seeing $(a_1,\ldots,a_n)$ as the first $n$ measurement outcomes is equal to 
\begin{equation}\label{eq:QMP} 
    p^{(n)}_{\omega,\rho}(a_1,\ldots,a_n) =  \tr [(\psi_{\theta^{n-1}\omega,a_n}\circ \cdots \circ \psi_{\omega,a_1})(\rho)] .
\end{equation}
The CPTP property ensures that this is indeed a probability measure on~$\cA^n$ and that the passage from $n$ to $n-1$ is consistent. Hence, by the Kolmogorov extension theorem, the measure~$p^{(n)}_{\omega,\rho}$ can be thought of as the $n$-th marginal of an $\omega$-dependent measure~$p_{\omega,\rho}$ on the space~$\cA^\nn$ of infinite sequences of measurement outcomes. 

Given these first $n$ outcomes, the state of the quantum system is described by the density matrix
$$ 
    \rho_n =  \frac{(\psi_{\theta^{n-1}\omega,a_n}\circ \cdots \circ \psi_{\omega,a_1})(\rho)}{p^{(n)}(a_1,\ldots,a_n)} .
$$
This means that, after $n$ steps, the average state is given by $(\varphi_{\theta^{n-1}\omega} \circ \dotsb \circ \varphi_{\omega})(\rho)$. The CPTP property ensures that, with probability one, this is indeed still an element of~$\mathbb{S}$.

To each outcome $a\in \cA$, we suppose associated a quantity $f_\omega(a) \in \rr$, which may depend on the disorder realization~$\omega$. Our main result is a quenched large deviation principle for the Birkhoff sums $f_\omega(a_1)+\cdots + f_{\theta^{n-1}\omega}(a_n),$ with respect to the above QMP. To this end, we study the moment generating functionals
\begin{equation}
\label{eq:mgf-def}
    M_{\omega,\rho}^{(n)}(\alpha) := \sum_{\vec{a}\in \cA^n} \Exp{-\alpha \sum_{j=1}^n f_{\theta^{j-1}\omega}(a_j)} p^{(n)}_{\omega,\rho}(a_1,\ldots,a_n) .
\end{equation}
In light of \eqref{eq:QMP}, $M_{\omega,\rho}^{(n)}$ can be expressed as follows
\begin{equation}
\label{eq:mgf-as-tr}
    M_{\omega,\rho}^{(n)}(\alpha) =  \tr \left[ \left(\varphi_{\theta^{n-1}\omega}^{(\alpha)} \circ \cdots \circ \varphi_{\omega}^{(\alpha)}\right)(\rho) \right],
\end{equation}
where
\begin{equation}
\label{eq:varphialpha}
    \varphi^{(\alpha)}_{\omega} =  \sum_{a \in \cA} \Exp{-\alpha f_{\omega}(a)} \psi_{\omega,a}.
\end{equation}
We will sometimes refer to $\varphi^{(\alpha)}_{\omega}$ as in \eqref{eq:varphialpha} as an \emph{analytic deformation} of~$\varphi_\omega$. The following fact is easy to verify: with the standing assumption that the original map $\varphi_\omega$ is CPTP and irreducible, the deformation $\varphi^{(\alpha)}_{\omega}$ is completely positive and irreducible for all $\alpha \in \rr$.
The growth rate of $M_{\omega,\rho}^{(n)}(\alpha)$ is given by the $\alpha$-dependent \emph{top Lyapunov exponent}
\begin{equation}
\label{eq:lambda-alpha}
    \lambda(\alpha) := \lim_{n\to\infty} \frac 1n \log \left\|\varphi^{(\alpha)}_{\theta^{n-1}\omega} \circ \dotsb \circ \varphi^{(\alpha)}_{\theta\omega} \circ \varphi^{(\alpha)}_{\omega}\right\|\opone .
\end{equation}
While $\pp$-almost sure existence of the limit as a deterministic quantity for every fixed~$\alpha$ is guaranteed on general grounds by ergodicity of~$(\Omega,\theta,\pp)$ and Kingman's theorem~\cite{Ki68}, we will need the following assumptions based on \cite{PS23} in order to obtain further desirable properties: 
\begin{enumerate}[label=\textbf{(A\arabic{*})}]
    \item\label{A1} 
    With
    \[ 
        v(\phi) := \inf \{ \| \phi(X)\|_1 \ : \ X \in \mathbb{S}\} 
    \]
    the function $\omega \mapsto \log v(\varphi_\omega^*)$ belongs to $L^1(\Omega,\dd\pp)$.
  \item\label{A2} 
    There exists $N_0 \in \nn$ such that 
    \[ 
        \pp\left \{\omega\in\Omega \ : \ \varphi_{\theta^{N_0-1}\omega} \circ \dotsb \circ \varphi_{\theta\omega} \circ \varphi_{\omega} \text{ is positivity improving} \right \} > 0.
    \]
    \item\label{A3} The random variable $F_\omega  :=  \max_{a\in\cA}|f_\omega(a)| $ satisfies $\Exp{\alpha F} \in L^1(\Omega,\dd\pp)$ for all~$\alpha \in \rr.$
\end{enumerate}

\begin{remark}
    Assumptions~\ref{A1} and~\ref{A2} are made at the level of the average maps $\varphi_\omega$ only, and not at the level of the instruments $(\psi_{\omega,a})_{a\in\cA}$ or the deformed maps $\varphi_\omega^{(\alpha)}$.
\end{remark}

With this notation, our main result is the following quenched large deviation principle, proved in Section~\ref{sec:proof}.

\begin{theorem}
\label{thm:main-abstract}
    If {\textnormal{\ref{A1}}}, {\textnormal{\ref{A2}}} and {\textnormal{\ref{A3}}} hold and $\lambda^*(s) =  \sup_{\alpha\in\rr} (s\alpha - \lambda(\alpha))$ is the Legendre transform of the Lyapunov exponents $\lambda(\alpha)$ in~\eqref{eq:lambda-alpha}, then, for $\pp$-almost all~$\omega$, we have
    \begin{subequations}
    \label{eq:ldp}
    \begin{align}
        -\inf_{s \in \operatorname{int} E} \lambda^*(s) 
            &\leq \liminf_{n\to\infty} \frac 1n \log p^{(n)}_{\omega,\rho} \left\{ (a_1,\dotsc,a_n) \ : \ \frac 1n \sum_{j=1}^n f_{\theta^{j-1} \omega} (a_j) \in E \right\} \\
            &\leq \limsup_{n\to\infty} \frac 1n \log p^{(n)}_{\omega,\rho} \left\{ (a_1,\dotsc,a_n) \ : \ \frac 1n \sum_{j=1}^n f_{\theta^{j-1} \omega} (a_j) \in E \right\} 
            \leq  -\inf_{s \in \operatorname{cl} E} \lambda^*(s) 
    \end{align}
    \end{subequations}
    for every Borel set $E\subseteq\rr$. 
\end{theorem}

For every fixed~$\alpha$, the results of~\cite{MS22,PS23} will apply and tell us\,---\,among other things\,---\,that there is a $\mathbb{S}^\circ$-valued random variable $Z^{(\alpha)}$ that satisfies the cocycle relation
$$
    \varphi_\omega^{(\alpha)} \left ( Z_{\theta^{-1}\omega}^{\alpha} \right ) =  Z_\omega^{(\alpha)} \tr \left [\varphi_\omega^{(\alpha)} \left ( Z_{\theta^{-1}\omega}^{\alpha} \right ) \right ]
$$ 
for $\pp$-almost all~$\omega$ and such that 
\begin{equation} \label{eq:proto-lambda-as-avg}
    \lambda(\alpha)
    =  \int_\Omega \log \tr\left[\varphi_{\omega}^{(\alpha)} \left ( Z_{\theta^{-1}\omega}^{(\alpha)} \right ) \right] \dd \pp(\omega) .
\end{equation}
For the derivation of the relation \eqref{eq:proto-lambda-as-avg} from Theorem~1 in~\cite{PS23}, see Lemma~\ref{lem:Lyap-identities} below.
These objects are at the heart of the proof of Theorem \ref{thm:main-abstract}, which has of three main steps and is carried out in Section~\ref{sec:proof}. We summarize the main steps of the proof here:
\begin{description}
    \item[Step 1.] In view of the specific form of the dependence of~$\varphi_\omega^{(\alpha)}$ in~$\alpha$, we use arguments from~\cite{MS22,PS23} to show that there is an event of full measure on which the limiting random variables $Z^{(\alpha)}$ (and analogous objects for adjoints) and the top Lyapunov exponent $\lambda(\alpha)$ exist simultaneously for every deformation parameter~$\alpha \in \rr$. This is the content of Lemmas~\ref{lem:simul-alpha} and~\ref{lem:simul-alpha-ii}, below.
    
    \item[Step 2.] We then show that $Z^{(\alpha)}$ and $\lambda(\alpha)$ possess some regularity properties as functions of the deformation parameter~$\alpha$, namely: 
    \begin{itemize}
        \item[2a.] 
        The map $\alpha \mapsto Z^{(\alpha)}$ is continuous. This is the content of Proposition~\ref{prop:continuous-Z}.
        \item[2b.] 
        The map $\alpha \mapsto \lambda(\alpha)$ is differentiable. This is the content of Proposition~\ref{prop:lambda-diff}.
    \end{itemize}
    \item[Step 3.] We deduce the large deviation principle from the G\"artner--Ellis theorem after relating the top Lyapunov exponent $\lambda(\alpha)$ to the cumulant-generating functions via 
     \begin{equation}
    \label{eq:Lyap-as-rcgf}
       \lambda(\alpha) =  \lim_{n\to\infty} \frac 1n \log
       M_{\omega,\rho}^{(n)}(\alpha) .
    \end{equation}
\end{description}

\begin{remark}
    As far as we are aware, previous mathematical works on random repeated measurements  were restricted to the setup of Example~\ref{ex:examples-intro}.i. 
    In the deterministic case $|\mathcal{S}|=1$, the analyses of~\cite[\S{4}]{CJPS19} and~\cite[\S{1.2}]{BCJP21} bypass~{\textnormal{\ref{A2}}} as long as $\rho \in \mathbb{S}^{\circ}$ and $\varphi_{\omega_0}(\rho) = \rho$. The method there does not involve differentiability in~$\alpha$ of the limit~\eqref{eq:Lyap-as-rcgf}, and instead relies on decoupling conditions and Ruelle--Lanford functions. 
    Our results encompass the cases where $1 < |\mathcal{S}|<\infty$, considered in~\cite{BB20,BJP22} with $\pp$ a Bernoulli or an irreducible Markov measure, respectively. While combining Lemma~2.3.iii and Theorem~4.2.vi in \cite{BJP22} bypasses~{\textnormal{\ref{A2}}} in the Markovian regime, it should be noted that the large deviation principle there obtained  concerns \emph{annealed} randomness of the environment, as opposed to quenched randomness. The proof relies on differentiability of an annealed analogue of the limit~\eqref{eq:Lyap-as-rcgf} where $M^{(n)}(\alpha)$ is integrated with respect to~$\pp$ \emph{inside} the logarithm.
\end{remark}

\begin{remark}
    Suppose that the distribution of $\omega\mapsto \varphi_\omega$ is supported on finitely many points, i.e.\ there are only finitely many possible channels.  Then the only way {\textnormal{\ref{A1}}} could fail would be if one of the possible channels satisfies $\ker \varphi_{\omega}^* \cap \mathbb{S} \neq \emptyset$. Such maps are called \emph{transient} in~\cite{PS23} and would contradict our standing irreducibility assumption; see Lemma~\ref{lem:irred-ker-star} below. Therefore, if there are finitely many possible channels, then~{\textnormal{\ref{A1}}} cannot fail. Similarly, in this case it is sufficient for~{\textnormal{\ref{A2}}} that one of the possible channels be positivity improving. However, we want to emphasize that this is not necessary, as shown by the example below.
\end{remark}

\begin{example}
    Consider the dimension $d=2$ and suppose that the only two possible CPTP maps are of the form 
    \begin{align*}
        \phi_i(\rho)\ &=\epsilon_i \gamma_i \tr(\rho) + (1-\epsilon_i) \rho
    \end{align*}
    where $\gamma_i \in \mathbb{S}\setminus\mathbb{S}^\circ$ (for $d=2$, this means pure) and $\epsilon_i \in (0,1)$. Such channels, taken individually, are not positivity improving. However, their composition is such that
    \begin{align*}
        \phi_2(\phi_1(\rho)) 
           &=\epsilon_2 \gamma_2 + \epsilon_1(1-\epsilon_2) \gamma_1 + (1-\epsilon_1)(1-\epsilon_2) \rho
    \end{align*}
    will be positive definite for all $\rho\in\mathbb{S}$ if the kernels of $\gamma_1$ and $\gamma_2$ intersect trivially, which is generic. This provides a setup where~\textnormal{\ref{A2}} holds without any individual channel being positivity improving.
\end{example}

\section{Application to entropy production}
\label{sec:EP}

Following e.g.~\cite[\S{3}]{HP13}\,---\,also see~\cite[\S{3.1}]{HJPR18}, \cite[\S{5.1}]{BB20} and~\cite[\S{1.4}]{BCJP21}\,---, there is a choice of instruments that holds a particular significance from the point of view of quantum thermodynamics. We use this section to investigate consequences of Theorem~\ref{thm:main-abstract} in this regard. 
This is related to the foundational progress from~\cite{Ku00,Ta00,TM05,JOPS12,BBJPP23} on the quantum counterparts of~\cite{ECM93,ES94,GC95,Cr99} in the two-time measurement framework. 

Suppose that we have measurable maps $\omega \mapsto U_\omega \in \operatorname{U}(\cc^d\otimes\cc^m)$ and $\omega \mapsto \xi_\omega$, where $\xi_\omega$ is a Gibbs state on~$\cc^m$ with spectral decomposition 
\[ 
    \xi_\omega = \sum_{\epsilon=1}^m \Exp{-\beta_\omega E_{\omega,\epsilon}}\pi_{\omega,\epsilon}
\]
using rank-one projectors, an increasing energy labeling $\epsilon \mapsto E_{\omega,\epsilon}$ and a positive inverse temperature~$\beta_\omega$.
    {The assumption that the eigenvalues of every corresponding Hamiltonian are all simple is made for ease of presentation. The argument could be easily adapted to Hamiltonians with degenerate eigenvalues, although if the degeneracy of eigenvalues depends non-trivially on~$\omega$, some care would need to be taken to adjust the argument to account for the fact that the number of possible measurement outcomes is then random, e.g.\ by adding fictitious outcomes of probability~0.}
    
Then, consider instruments labeled by pairs $(\epsilon,\epsilon') \in \{1,2,\dotsc,m\}^2$ of indices for the spectral decomposition, that are of the form
\[
    \psi_{\omega,(\epsilon,\epsilon')}(\rho) = \tr[\pi_{\omega, \epsilon'}(U_\omega (\rho \otimes \pi_{\omega,\epsilon}\xi_\omega) U_\omega^*)] .
\]
The quantity
\[
    f_{\omega}(\epsilon,\epsilon') 
    =
    \beta_\omega(E_{\omega,\epsilon'}-E_{\omega,\epsilon}) 
\]
has a natural physical interpretation. There is a quantum system and a random (i.e.\ $\omega$-dependent) sequence of probes, initially in Gibbs states. Sequentially, a probe arises, we measure its energy, it interacts with the quantum system through a random unitary evolution, we measure its energy again, and it leaves forever. We record the resulting energy differences, rescaled by the appropriate temperature.
Our main result provides the large-deviation principle for 
\[
    \Sigma_{\omega,n}((\epsilon_1,\epsilon'_1), (\epsilon_2,\epsilon'_2), \dotsc, (\epsilon_n,\epsilon'_n))
    =  \sum_{j=0}^{n-1} \beta_{\theta^j\omega}(E_{{\theta^j\omega},\epsilon'_j}-E_{{\theta^j\omega},\epsilon_j}) ,
\]
i.e.\ for the \emph{Clausius entropy change} of the chain of Gibbs states, provided that one can verify (A1)--(A3). Note that, in terms of the properties of the probes, (A3) amounts to a lower bound on the temperature and an upper bound on the width of the energy spectrum. 

There is another, more information theoretic notion of \emph{entropy production} that is relevant to such setups, but its definition requires additional structures.
We make the following assumption of \emph{time-reversal invariance} (TRI) at the level of the driving ergodic process and at the level of the microscopic quantum dynamics:
\begin{description}
    \item[TRI Assumption] We suppose that there exists a measurable involution~$\mathsf{T}:\Omega \to \Omega$ such that 
    \begin{itemize}
        \item $\mathsf{T} \circ \theta  = \theta^{-1} \circ \mathsf{T}$,
        \item $\pp \circ \mathsf{T}^{-1} = \pp$,
    \end{itemize}
    and antiunitary involutions~$\tau : \cc^d \to \cc^d$ and $\tau' : \cc^m \to \cc^m$ such that 
    \begin{itemize}
        \item $(\tau \otimes \tau')U_\omega= U_{\mathsf{T}\omega}^* (\tau \otimes \tau')$,
        \item $\tau' \xi_{\omega}= \xi_{\mathsf{T}\omega} \tau' $.
    \end{itemize}
\end{description}

\begin{example}
\label{rem:TRI}
    We can discuss TRI in the context of Example~\ref{ex:examples-intro}.
    \begin{enumerate}
        \item[i.] In the setup of Example~\ref{ex:examples-intro}.i, the map $\mathsf{T} : (\omega_n)_{n\in\zz} \mapsto (\omega_{-n})_{n\in\zz}$ clearly satisfies the condition on the interplay of~$\mathsf{T}$ and $\theta$, and the condition on the interplay of~$\pp$ and~$\mathsf{T}$ is then naturally interpreted as a generalization of detailed balance for Markov chains. 
        With $\omega \mapsto U_\omega \in \operatorname{U}(\cc^d\otimes\cc^m)$ and $\omega \mapsto \xi_\omega$ depending only on~$\omega_0$, the conditions on $\tau$ and $\tau'$ further simplify to $(\tau \otimes \tau')U_\omega = U^*_{\omega} (\tau \otimes \tau')$ and $\tau' \xi_{\omega} = \xi_{\omega} \tau'$. Writing $U_\omega = \exp(\mathrm{i} t_\omega H_\omega)$, note that such $\tau$ and $\tau'$ exist as complex conjugations if there exist orthonormal bases of~$\cc^d$ and~$\cc^m$ (used to construct an orthonormal basis of~$\cc^d \otimes \cc^m$) with respect to which $\xi_{\omega}$ and $H_\omega$ have real matrix elements.
        
        \item[ii.] In the setup of Example~\ref{ex:examples-intro}.ii, the map $\mathsf{T}: (q,p) \mapsto (q,-p)$ that changes sign of the  momentum variables satisfies the condition on the interplay of~$\mathsf{T}$ and $\theta$, and the condition on the interplay of~$\mathsf{T}$ and $\pp$ is then naturally interpreted as a symmetry of the initial distribution of the momentum variables.
    \end{enumerate}
\end{example}

\begin{remark}
\label{rem:xi-labeling}
    Fix $\omega$ and suppose that $ \xi_{\omega} = \tau'\xi_{\mathsf{T}\omega} \tau'$ as required in TRI. Comparing the spectral decompositions following the above labeling conventions, we find that  
    $\pi_{\mathsf{\omega},\epsilon}= \tau'\pi_{\mathsf{T}\omega,\epsilon}\tau'$ for each $\epsilon$.
\end{remark}

\begin{thought}
    Consider the setup of Remark~\ref{rem:TRI}.i and suppose that we are given an ordered set of $n$ observations, indexed by $j$, and with the following content:
    \begin{itemize}
        \item the provenance of the $j$-th probe (and hence its temperature, Hamiltonian, etc.);
        \item the values of energy measurements of the $j$-th probe before and after it interacted with the system of interest.
    \end{itemize}
    However, we do not know if the ordered set of observations has been reversed ($j \leftrightarrow n+1-j$) before it was given to us, and we are asked to place a bet on either of those two hypotheses.
\end{thought}

To perform well in this thought experiment, the Neyman--Pearson lemma suggests comparing the likelihood ratio
\[
    \frac{
    \pp_n[\omega_0,\omega_1,\dotsc,\omega_n] p^{(n)}_{(\omega_0, \omega_1, \dotsc, \omega_{n-1}),\rho}((\epsilon_1,\epsilon'_1), (\epsilon_2,\epsilon'_2), \dotsc, (\epsilon_n,\epsilon'_n))}
    {\pp_n[\mathsf{T}\omega_{n-1},\mathsf{T}\omega_{n-2},\dotsc,\mathsf{T}\omega_0] p^{(n)}_{(\mathsf{T} \omega_{n-1}, \mathsf{T}\omega_{n-2}, \dotsc, \omega_0),\rho} ((\epsilon'_n,\epsilon_n), (\epsilon'_{n-1},\epsilon_{n-1}), \dotsc, (\epsilon'_1,\epsilon_1))}
\]
to some threshold; see e.g.~\cite[\S{7.1}]{DeZe} for abstract theory and~\cite[\S{2.9}]{BJPP18} for an application to repeated quantum measurements.

Under the TRI Assumption, the $n$-th marginals of $\pp$ at the numerator and denominator cancel each other out, and that motivates the other notion of entropy production we are about to discuss:
the random variable
\begin{equation}
\label{eq:info-EP}
    \sigma_{\omega,\rho,n}((\epsilon_1,\epsilon'_1), (\epsilon_2,\epsilon'_2), \dotsc, (\epsilon_n,\epsilon'_n)) :=  \log \frac
        {p^{(n)}_{\omega,\rho}((\epsilon_1,\epsilon'_1), (\epsilon_2,\epsilon'_2), \dotsc, (\epsilon_n,\epsilon'_n))}
        {p^{(n)}_{\mathsf{T}(\theta^{n-1} \omega),\rho} ((\epsilon'_n,\epsilon_n), (\epsilon'_{n-1},\epsilon_{n-1}), \dotsc, (\epsilon'_1,\epsilon_1))}
\end{equation}
is a log-likelihood ratio that can be interpreted as an information-theoretic measure of irreversibility of the sequence of measurement outcomes~\cite{JOPS12,BJPP18}. 
The following lemmas are minor adaptations of e.g.~\cite[App.~C]{HJPR18} or~\cite[\S{5.4}]{BJP22} to our slightly more intricate notion of TRI. We still provide the proofs in Section~\ref{sec:EP-proofs} for completeness.
\begin{lemma}
\label{lem:dual-1-minus}
    Under the TRI Assumption, we have the symmetry  $\varphi_{\mathsf{T}\omega}^{(\alpha)} (\tau X \tau) = \tau (\varphi_{\omega}^{(1-\alpha)})^*(X) \tau$ 
    for all~$X\in\mathbb{S}$,~$\alpha\in\cc$ and $\omega\in\Omega$.
\end{lemma}

\begin{lemma}
\label{lem:first-EP-bound}
    With 
    \[ 
     \Delta_\rho = \frac{1}{\inf\sp \rho},
    \]
    we have 
     \begin{align*}
      \Exp{\Sigma_{\omega,n}} \Delta_\rho^{-1} {p}^{(n)}_{\mathsf{T}\theta^{n-1}\omega,\rho} \leq p^{(n)}_{\omega,\rho}           &\leq\Exp{\Sigma_{\omega,n}} \Delta_\rho {p}^{(n)}_{\mathsf{T}\theta^{n-1}\omega,\rho} .
    \end{align*}
\end{lemma}

Hence, the abstract large-deviation principle provides, in the present context, a route to the large-deviation principle for this information-theoretic notion of entropy production as well. The following theorem is proved in Section~\ref{sec:EP-proofs}.

\begin{theorem}
\label{thm:main-EP}
    Suppose that {\textnormal{\ref{A1}}}--{\textnormal{\ref{A3}}} and the TRI Assumption hold.
    Then, $\pp$-almost surely, the large deviation principle holds for the sequence $(\sigma_{\omega,\rho,n})_{n\in\nn}$ be defined according to~\eqref{eq:info-EP}, with a convex rate function~$J$ that is independent from~$\rho \in \mathbb{S}^\circ$ and $\omega$.
    Moreover, the rate function~$J$ satisfies the Gallavotti--Cohen symmetry $J(-s) = J(s) + s$ for all $s\in\rr$.
\end{theorem}

\begin{remark}
    It has already been observed in the literature that the property of this two-time measurement framework that the information-theoretic notion of entropy production is exponentially equivalent to a Birkhoff-like sum along the QMP (i.e.\ the Clausius entropy change) is an extraordinarily simplifying feature. Indeed, as seen through examples~\cite{BCJP21}, the study of the information-theoretic notion of entropy production can be considerably more intricate outside of the two-time measurement framework; also see~\cite{CJPS19}.
\end{remark}

\section{Proof of Theorem~\ref{thm:main-abstract}}
\label{sec:proof}

\subsection{Notation and preliminary lemmas}

Before we proceed with the proof of Theorem~\ref{thm:main-abstract}, let us introduce some notation and collect some preliminary results. For $n \geq 0$, we set
\[ 
    \Phi_{n,\omega}^{(\alpha)} = \varphi^{(\alpha)}_{\theta^{n-1}\omega} \circ \dotsb \circ \varphi^{(\alpha)}_{\theta\omega} \circ \varphi^{(\alpha)}_{\omega}
\]
and 
\[ 
    \Phi_{-n,\omega}^{(\alpha)} = \varphi^{(\alpha)}_{\omega} \circ \dotsb \circ \varphi^{(\alpha)}_{\theta^{-n+2}\omega} \circ \varphi^{(\alpha)}_{\theta^{-n+1}\omega} .
\]
We will use ``$\,\cdot\,$'' for the \emph{projective action} of linear maps on~$\mathbb{S}$:
\begin{equation}\label{eq:projectiveaction}
    \phi \cdot X := \frac{\phi(X)}{\tr \phi(X)} .
\end{equation}
If $\phi(X)=0$ the projective action as in \eqref{eq:projectiveaction} is undefined.  However, since $\varphi$ is trace preserving we note $\varphi^{(\alpha)}\cdot \rho$ is well defined for any density matrix $\rho$, since $\tr \varphi^{(\alpha)}(\rho) \ge \Exp{-|\alpha|F_\omega} $. Similarly the following lemma shows that the projective action of $\varphi^*$ (and therefore) $\varphi^{(\alpha)*}$ is well defined.

\begin{lemma}
\label{lem:irred-ker-star}
    If a map~$\phi$ is irreducible and trace preserving, then $\ker \phi^* \cap \mathbb{S} = \emptyset$.
\end{lemma}

\begin{proof}
    If~$\phi$ is irreducible and trace preserving, then $\phi^*$ is irreducible and unital. By irreducibility, the map $(\operatorname{id}_{\cB(\cc^d)} + \phi^*)^{d-1}$ is positivity improving, so any element of $\ker \phi^* \cap \mathbb{S}$ must actually be in~$\mathbb{S}^\circ$. However, if $X \in \mathbb{S}^\circ$, then $X \geq \eta \one$ for some $\eta > 0$ and
    unitality implies that 
    \begin{align*}
        \phi^*(X) \geq \phi^*(\eta\one) = \eta \phi^*(\one) = \eta \one ,
    \end{align*}
    so $X$ cannot be in $\ker \phi^*$. We conclude that $\ker \phi^* \cap \mathbb{S} = \emptyset$.
\end{proof}

\begin{lemma}
\label{lem:A-to-alpha}
    If~\textnormal{\ref{A1}} and~\textnormal{\ref{A3}} hold, then
    \[
        \int_\Omega \left |\log \|(\varphi_{\omega}^{(\alpha)})^* \|\opone \right | + \left | \log v((\varphi_{\omega}^{(\alpha)})^* ) \right |\dd\pp(\omega)< \infty 
    \]
    and 
    \[
        \int_\Omega \left | \log \|\varphi_{\omega}^{(\alpha)} \|\opone \right | + \left | \log v(\varphi_{\omega}^{(\alpha)} ) \right | \dd\pp(\omega) < \infty
    \]
    for all $\alpha \in \rr$.
\end{lemma}

\begin{proof}
    We start with $\alpha = 0$. Because $\varphi_\omega$ is CPTP, we have $\|\varphi_\omega\|\opone = 1$ and $v(\varphi_\omega) = 1$. For the adjoint, we will use the fact that $\varphi^{*}_\omega$ is then positive and unital. A lower bound is easily obtained: $\|\varphi^*_\omega\|\opone \geq \|\varphi^*_\omega(\one)\|_1/\|\one\|_1 = \|\one\|_1/\|\one\|_1 = 1$. To obtain an upper bound, we note that $\|\varphi^*_\omega\|\opone \leq d \|\varphi^*\|\opty = d\|\varphi\|\opone=d$.  Finally, $\log v(\varphi_\omega^*)\in L^1$ by Assumption~\ref{A1}. 

    Turning now to $\alpha \neq 0$, observe that due to~\ref{A3} we have 
    \begin{equation}
    \label{eq:basic-F-bound}
        \Exp{-F_\omega|\alpha|} \varphi_\omega(X) \leq \varphi_\omega^{(\alpha)}(X) \leq \Exp{F_\omega|\alpha|} \varphi_\omega(X)
    \end{equation} 
    for every $\alpha \in \rr$ and~$X \in \mathbb{S}$, with $F \in L^1(\Omega,\dd\pp)$.
        % \footnote{While this proof does not exploit the full strength of~\ref{A3}, some later ones will.}
    Hence, the case of an arbitrary $\alpha \in \rr$ follows from our analysis of the case $\alpha=0$.
\end{proof}

At {fixed} $\alpha\in\rr$, Assumption~\ref{A2} and the bounds in~\eqref{eq:basic-F-bound} guarantee that there exist random variables $Z^{(\alpha)} : \Omega \mapsto \mathbb{S}^{\circ}$ and $\Zp^{(\alpha)}: \Omega \mapsto \mathbb{S}^{\circ}$ satisfying the identities 
\begin{equation}
\label{eq:cocycle}
    \varphi_{\omega}^{(\alpha)} \cdot Z_{\theta^{-1}\omega}^{(\alpha)} = Z_{\omega}^{(\alpha)}
    \qquad\text{and}\qquad 
    (\varphi_{\omega}^{(\alpha)})^* \cdot \Zp_{\theta\omega}^{(\alpha)} = \Zp_{\omega}^{(\alpha)}
\end{equation} 
for $\pp$-almost every~$\omega$. This is the content of Theorem~1 in~\cite{MS22}.
If, in addition, \ref{A1} and~\ref{A3} hold, then we have the following alternate expression for the Lyapunov exponent $\lambda(\alpha)$:
\begin{equation}
\label{eq:PS23lyap}
    \lambda(\alpha) 
    = \int_\Omega \log \tr[(\varphi_{\omega}^{(\alpha)})^*( \Zp_{\theta\omega}^{(\alpha)})] \dd \pp(\omega),
\end{equation}
and\,---\,still at fixed~$\alpha$\,---\,
\begin{equation}
\label{eq:norm-vs-application-to-rho}
    \lambda(\alpha) = \lim_{n\to\infty} \frac 1n \log \tr \left[ \left(\varphi_{\theta^{n-1}\omega}^{(\alpha)} \circ \cdots \circ \varphi_{\omega}^{(\alpha)}\right)(\rho)\right]
\end{equation}
for $\pp$-almost all~$\omega$ and all $\rho \in\mathbb{S}$. 
This is the content of Theorem~1 in~\cite{PS23}. 

\begin{lemma}\label{lem:Lyap-identities} The following identities hold for the Lyapunov exponents:
\begin{equation}\label{eq:Lyap-not-adjoint}
    \lambda(\alpha) 
    =  \int_\Omega \log \tr[\varphi_{\omega}^{(\alpha)}( Z_{\theta^{-1}\omega}^{(\alpha)})] \dd \pp(\omega) ,
\end{equation}
and 
\begin{equation}
\label{eq:Lyap-as-average}
    \lambda(\alpha) = \int_\Omega \left( \log  \tr [\Zp_{\theta\omega}^{(\alpha)} \varphi_{\omega}^{(\alpha)} (Z_{\theta^{-1}\omega}^{(\alpha)})] - \log \tr [\Zp_{\theta\omega}^{(\alpha)} Z_{\omega}^{(\alpha)}] \right )\dd\pp(\omega).
\end{equation}
\end{lemma}
\begin{proof} 
Note that the cocycle relations \eqref{eq:cocycle} imply that
\begin{equation}\label{eq:lyap-cocycle}
\log \tr [ \varphi_{\omega}^{(\alpha)}(Z_{\theta^{-1}\omega}^{(\alpha)})] = \log  \tr [\Zp_{\theta\omega}^{(\alpha)} \varphi_{\omega}^{(\alpha)} (Z_{\theta^{-1}\omega}^{(\alpha)})] - \log \tr[ \Zp_{\theta\omega}^{(\alpha)} Z_{\omega}^{(\alpha)}]
\end{equation}
and
\begin{align*} 
\log \tr [ (\varphi_{\omega}^{(\alpha)})^*(\Zp_{\theta\omega}^{(\alpha)})] 
    &= \log  \tr [(\varphi_{\omega}^{(\alpha)})^* (\Zp_{\theta\omega}^{(\alpha)}) Z_{\theta^{-1}\omega}^{(\alpha)}] - \log \tr[ \Zp_{\omega}^{(\alpha)} Z_{\theta^{-1}\omega}^{(\alpha)}] \\ 
    &=  \log  \tr [\Zp_{\theta\omega}^{(\alpha)} \varphi_{\omega}^{(\alpha)} (Z_{\theta^{-1}\omega}^{(\alpha)})] - \log \tr[ \Zp_{\omega}^{(\alpha)} Z_{\theta^{-1}\omega}^{(\alpha)}] .
\end{align*}
Therefore,
\begin{equation}
    \log \tr [ \varphi_{\omega}^{(\alpha)}(Z_{\theta^{-1}\omega}^{(\alpha)})] - \log \tr [ (\varphi_{\omega}^{(\alpha)})^*(\Zp_{\theta^\omega}^{(\alpha)})] = \log \tr[ \Zp_{\omega}^{(\alpha)} Z_{\theta^{-1}\omega}^{(\alpha)}] - \log \tr[ \Zp_{\theta\omega}^{(\alpha)} Z_{\omega}^{(\alpha)}] .
\end{equation}
Noting that the right-hand side is of the form $g-g\circ\theta$ and that the terms $\log \tr [ \varphi_{\omega}^{(\alpha)}(Z_{\theta^{-1}\omega}^{(\alpha)})]$ and $\log \tr [ (\varphi_{\omega}^{(\alpha)})^*(\Zp_{\theta^\omega}^{(\alpha)})]$ are in $ L^1(\Omega,\dd \pp)$ by Lemma \ref{lem:A-to-alpha}, we conclude from Lemma \ref{lem:shift-cancellation} below that 
\begin{equation}
\label{eq:integral-Lyap-cancellation}
   \int_\Omega  \log \tr [ \varphi_{\omega}^{(\alpha)}(Z_{\theta^{-1}\omega}^{(\alpha)})] - \log \tr [ (\varphi_{\omega}^{(\alpha)})^*(\Zp_{\theta^\omega}^{(\alpha)})] \dd \pp(\omega)  =  0 .
\end{equation}
The identity~\eqref{eq:Lyap-not-adjoint} follows from~\eqref{eq:PS23lyap} and~\eqref{eq:integral-Lyap-cancellation}, and the identity~\eqref{eq:Lyap-as-average} then follows from~\eqref{eq:Lyap-not-adjoint} and~\eqref{eq:lyap-cocycle}.
\end{proof}

\subsection{Reduction appealing to the G\"artner--Ellis theorem}

We want to prove that, for $\pp$-almost all~$\omega$, the large deviation principle~\eqref{eq:ldp} holds. In view of the G\"artner--Ellis theorem~\cite[\S{II.6}]{Ell}, this reduces to showing the following about cumulant-generating functionals: 
\begin{quote}
    For $\pp$-almost all~$\omega$,
\begin{equation}
\label{eq:lambda-is-lim-cgf}
    \lim_{n\to\infty} \frac 1n \log
    M_{\omega,\rho}^{(n)}(\alpha) = \lambda(\alpha)
\end{equation}
for all~$\alpha \in \rr$, and the function~$\alpha \mapsto \lambda(\alpha)$ is differentiable.
\end{quote} 
We have already seen that
\begin{equation*}
    M_{\omega,\rho}^{(n)}(\alpha) = \tr \left[ \left(\varphi_{\theta^{n-1}\omega}^{(\alpha)} \circ \cdots \circ \varphi_{\omega}^{(\alpha)}\right)(\rho)\right]
\end{equation*}
and that, under \ref{A1}--\ref{A3}, for fixed~$\alpha$, the relation~\eqref{eq:norm-vs-application-to-rho} holds
for $\pp$-almost all~$\omega$.  Combining these two equalities indeed yields the equality~\eqref{eq:lambda-is-lim-cgf}, but with the caveat that the set of full $\pp$-measure of~$\omega$ for which it holds depends\,---\,at least a priori\,---\,on~$\alpha$.

We are thus left with problems of two types: guaranteeing that almost sure results concerning objects at fixed~$\alpha \in \rr$ can be promoted to  results that, almost surely, hold simultaneously for all~$\alpha \in\rr$; and then establishing regularity of the dependence in~$\alpha$ of these objects.

\subsection{Simultaneous existence of limiting objects}
\label{ssec:simult}

As mentioned above, our first technical obstacle is that the full $\pp$-measure sets of $\omega$ for which~\eqref{eq:lambda-alpha} holds and for which~$Z_{\omega}^{(\alpha)}$, $\Zp_{\omega}^{(\alpha)}$,~$Z_{\theta\omega}^{(\alpha)}$ and~$\Zp_{\theta\omega}^{(\alpha)}$ are well defined and satisfy~\eqref{eq:cocycle} could a priori depend on~$\alpha$. This section is devoted to verifying that there is a common full $\pp$-measure set on which these properties and other intermediate results of~\cite{MS22,PS23} hold simultaneously for all~$\alpha \in \rr$.
Some of these intermediate results refer to a particular metric~$\dd$ on~$\mathbb{S}$. This metric is defined and then related to the trace norm in the following lemma.

\begin{lemma}
\label{lem:d-to-tr}
    Set
    \[
      m(X,Y) := \sup\{\lambda \geq 0 \ : \ \lambda X \leq Y \}  
    \]
    and 
    \[
        \dd(X,Y) := \frac{1-m(X,Y)m(Y,X)}{1+m(X,Y)m(Y,X)}.
    \]
    If $Y' \in \mathbb{S}$ is positive definite, then 
    \[
       \frac{1}{2}\|Y-Y'\|_1 \leq \dd(Y,Y')
            \leq \frac{1}{\min\operatorname{sp}(Y')} \|Y-Y'\|_1 
    \]
    for all~$Y \in \mathbb{S}$.
\end{lemma}

\begin{proof}
    Let $Y,Y'$ be as in the statement of the lemma. The first inequality is established as Lemma~3.8 in~\cite{MS22}, so we need only prove the second. Let $\xi=\|Y-Y'\|_1$ and $\eta=1/\min \operatorname{sp}(Y')$. We aim to show $\dd(Y,Y') \le \xi\eta$. The result is trivial if $\xi\eta\ge 1$ (since $\dd(Y,Y')\le 1$ in any case); so we may assume $\xi \eta < 1$.  Since $-\xi\one \leq Y-Y' \leq \xi\one$ and $\eta Y' \ge \one $, we have
    $$
        (1-\xi \eta) Y' \leq Y \le (1+\xi \eta) Y'.
    $$
    Thus $m(Y,Y')\geq \frac 1{1+\xi\eta}$, $m(Y',Y) \geq 1-\xi\eta>0$, and
    \[
        \dd(Y,Y')
           \leq \frac{1 - \frac{1-\xi\eta}{1+\xi\eta}}{1 + \frac{1-\xi\eta}{1+\xi\eta}} = \xi \eta.
          \qedhere 
    \]
\end{proof}

In the proof of Theorem~1 in~\cite{MS22}, the notion of \emph{contraction coefficient} of the projective action of a map with respect to this metric plays an important role. This is defined as
\begin{align*}
    c(\Phi) &:= \operatorname{diam}_{\dd}(\Phi\cdot\mathbb{S}) \\
    &\phantom{:}= \sup\left\{\dd(\Phi\cdot X, \Phi\cdot Y): X,Y \in \mathbb{S} \right\}.
\end{align*}
More precisely, we will be interested in 
\begin{equation}\label{eq:kappa}
    \log \kappa(\alpha)  := \inf_{N\in\nn} \frac{1}{N} \int \log c(\Phi_{N}^{(\alpha)}) \dd\pp.
\end{equation}
By Lemmas 3.12--3.14 in~\cite{MS22}, for each $\alpha\in\rr$, the number $\kappa(\alpha)$ lies in the interval~$[0,1)$ and there is a set of full $\pp$-measure on which
\begin{equation}\label{eq:ctokappa}
    \lim_{N\to \infty} c(\Phi_{N;\omega}^{(\alpha)})^{1/N} = \lim_{N\to \infty} c(\Phi_{-N;\omega}^{(\alpha)})^{1/N} = \kappa(\alpha).
\end{equation}
Furthermore, for each $Y\in \mathbb{S}$ and $N\ge 0$, 
\begin{equation}\label{eq:fundbound}
    \dd(\Phi_{-N;\omega}^{(\alpha)}\cdot Y, Z_\omega^{(\alpha)}) \le c(\Phi_{-N;\omega}^{(\alpha)}) \quad \text{and} \quad \dd((\Phi_{N;\omega}^{(\alpha)})^*\cdot Y, \Zp_\omega^{(\alpha)}) \le c(\Phi_{N;\omega}^{(\alpha)}).
\end{equation}

\begin{lemma}
\label{lem:simul-pos-def}
    If \textnormal{\ref{A2}} and~\textnormal{\ref{A3}} hold, then there exists a full $\pp$-measure set $\Omega_0$ with the following property: for all $\alpha \in \rr$ and $\omega \in \Omega_0$, there exists $N$ such that $\Phi_{\pm n,\omega}^{(\alpha)}$ and $(\Phi_{\pm n,\omega}^{(\alpha)})^*$ are positivity improving for all~$n \geq N$.
\end{lemma}

\begin{proof}[Proof sketch]
    The bounds~\eqref{eq:basic-F-bound} in the proof of Lemma~\ref{lem:A-to-alpha} show that the desired positivity-improving properties are independent of~$\alpha$. Hence, the proof boils down to the case $\alpha = 0$, for which one can use ergodicity, Lemma~\ref{lem:irred-ker-star} and the following facts about positivity:
        if $\Phi$ is positivity improving, then $\Phi^*$ is positivity improving;
        if $\Phi$ is positivity improving and $\psi$ is positive, then $\Phi \circ \psi$ is positivity improving;
        and
        if $\Phi$ is positivity improving and $\psi$ is positive with $\ker \psi^* \cap \mathbb{S} = \emptyset$, then $\psi \circ \Phi$ is positivity improving.
    We refer the reader to~\cite[\S{3.3}]{MS22} for more details.
\end{proof}

\begin{lemma}\label{lem:lyapunovlemma}
    If~\textnormal{\ref{A2}} and~\textnormal{\ref{A3}} hold, then there exists a set $\Omega_0$ of full $\pp$-measure with the following property: for all $\omega \in \Omega_0$, the convergences~\eqref{eq:lambda-alpha} and~\eqref{eq:norm-vs-application-to-rho} hold simultaneously for all~$\alpha\in\rr$.
\end{lemma}

\begin{proof}
    Let 
    \begin{equation}
        G_{N,\omega}(\alpha)  := \frac{1}{N} \log \left \| \Phi_{N,\omega}^{(\alpha)} \right \|_{\mathrm{op},1}
    \end{equation}
    and 
    \begin{equation}
        G'_{N,\rho,\omega}(\alpha) := \frac 1N \log \tr \left[ \Phi_{N,\omega}^{(\alpha)} (\rho)\right].
    \end{equation}
    We have already established that, for every fixed~$\alpha$, we have $G_N(\alpha) \to \lambda(\alpha)$ and $G'_{N,\rho}(\alpha) \to \lambda(\alpha)$ almost surely as $N\to\infty$. Hence, in view of Lemma~\ref{lem:convexlemma} below, it suffices to establish that, for every fixed $N$ and~$\rho$, the functions $\alpha \mapsto G_{N,\omega}(\alpha)$ and $\alpha \mapsto G'_{N,\rho,\omega}(\alpha)$ are convex. 
    
    Because $G_{N,\omega}(\alpha) = \sup_{\rho \in \mathbb{S}} G'_{N,\rho,\omega}(\alpha)$ by the Russo--Dye theorem and monotonicity of the logarithm, the problem further reduces to showing that $G'_{N,\rho,\omega}$ is convex. 
    But \eqref{eq:mgf-def}--\eqref{eq:mgf-as-tr} show that $G'_{N,\rho,\omega}$ is 
    a multiple of a cumulant-generating function and is therefore convex\,---\,this is a well-known consequence of H\"older's inequality.
\end{proof}

\begin{lemma}
\label{lem:simul-alpha}
    If {\textnormal{\ref{A2}}} and~{\textnormal{\ref{A3}}} hold, then there exists a set $\Omega_0$ of full $\pp$-measure with the following property: for all $\omega\in \Omega_0$ and all $\alpha \in \rr$, 
    \begin{equation}
    \label{eq:cbelowkappa}
        \limsup_{N\to \infty} c(\Phi_{N;\omega}^{(\alpha)})^{1/N} \leq \kappa(\alpha) \quad \text{and} \quad  \limsup_{N\to \infty} c(\Phi_{-N;\omega}^{(\alpha)})^{1/N} \leq \kappa(\alpha) .
    \end{equation}
\end{lemma}

\begin{proof}
    For every~$\alpha \in \rr$ and $\omega \in \Omega$, the contraction coefficients are submultiplicative by~\cite[\S{3.2}]{MS22} in the sense that
    \[
        c(\Phi_{N+M;\omega}^{(\alpha)}) \leq c(\Phi_{M;\theta^N\omega}^{(\alpha)}) c(\Phi_{N;\omega}^{(\alpha)}).
    \]
    Therefore, the inequalities in~\eqref{eq:cbelowkappa} follow from the continuity bound of Lemma~\ref{lem:equiconlemma} below and an abstract convergence result for continuous submultiplicative families stated and proved in Appendix~\ref{app:together} (Lemma~\ref{lem:appendixlemma}).
\end{proof}

\begin{lemma}\label{lem:equiconlemma}
    There exists a set $\Omega_0$ of full $\pp$-measure such that
    \begin{equation}
    \label{eq:equiconlemma}
        \left| c(\Phi_{n,\omega}^{(\alpha)})-c(\Phi_{n,\omega}^{(\beta)}) \right| \leq 2 |\beta-\alpha| \sum_{m=0}^{n-1} F_{\theta^m\omega}
    \end{equation}
    for all $\omega\in\Omega_0$, $\alpha,\beta\in \rr$ and $n\in \nn$.
\end{lemma}

\begin{proof}
    Let $\tilde{m}(X,Y)=m(X,Y)m(Y,X)$. Then, by \cite[Lemma 3.3]{MS22},
    \begin{equation}
    \label{eq:mtildeidentity}
    \begin{split}
        \tilde{m}(\Phi_n^{(\beta)} \cdot X, \Phi_n^{(\alpha)}\cdot X ) 
        &= \inf_{A,B\in \ss^\circ } \frac{\tr[A \Phi_n^{(\beta)}\cdot X ]  \tr[B \Phi_n^{(\alpha)}\cdot X]}{\tr[A\Phi_n^{(\alpha)}\cdot X] \tr[B\Phi_n^{(\beta)}\cdot X]} \\
        &= \inf_{A,B\in \ss^\circ } \frac{\tr[A \Phi_n^{(\beta)}(X)]  \tr[B \Phi_n^{(\alpha)}(X)]}{\tr[A\Phi_n^{(\alpha)}(X)] \tr[B\Phi_n^{(\beta)}(X)]},
    \end{split}
    \end{equation}
    where due to the homogeneity of the expression, the projective actions can be replaced by linear actions. By the definition of~$F$ in~\ref{A3} and countable additivity of probability, there exists a full $\pp$-measure set~$\Omega_0$ with the property that, for all~$\omega \in \Omega_0$, we have
    \begin{equation}\label{eq:trAphiXbound}
        \Exp{-|\beta-\alpha|\sum_{m=0}^{n-1} F_{\theta^m\omega}} \tr[ A \Phi_{n,\omega}^{(\alpha)}(X)] \leq \tr[A \Phi_{n,\omega}^{(\beta)}(X)] \leq 
        \Exp{|\beta-\alpha|\sum_{m=0}^{n-1} F_{\theta^m\omega}} \tr [ A \Phi_{n,\omega}^{(\alpha)}(X)],
    \end{equation}
    and similarly with $A\mapsto B$, for all $\alpha,\beta\in \rr$.
    Now fix $\alpha,\beta \in \rr$ and $\omega \in \Omega_0$. Combining \eqref{eq:mtildeidentity} and \eqref{eq:trAphiXbound}, we find that
    \begin{equation}
        \tilde{m}(\Phi_{n,\omega}^{(\beta)} \cdot X, \Phi_{n,\omega}^{(\alpha)}\cdot X ) \ge \Exp{-2|\beta-\alpha|\sum_{m=0}^{n-1} F_{\theta^m\omega}}  .
    \end{equation}
    Thus 
    \begin{equation}
        \dd(\Phi_{n,\omega}^{(\beta)} \cdot X, \Phi_{n,\omega}^{(\alpha)} \cdot X) \leq \frac{1 - \Exp{-2|\beta-\alpha|\sum_{m=0}^{n-1} F_{\theta^m\omega}}}{1 + \Exp{-2|\beta-\alpha|\sum_{m=0}^{n-1} F_{\theta^m\omega}}} \leq |\beta-\alpha|\sum_{m=0}^{n-1} F_{\theta^m\omega} .
    \end{equation}
    Hence
    \begin{equation} 
    \begin{split} \dd(\Phi_{n,\omega}^{(\beta)}\cdot X, \Phi_{n,\omega}^{(\beta)}\cdot Y) &\le \dd(\Phi_{n,\omega}^{(\beta)}\cdot X, \Phi_{n,\omega}^{(\alpha)}\cdot X) + \dd(\Phi_{n,\omega}^{(\alpha)}\cdot X, \Phi_{n,\omega}^{(\alpha)}\cdot Y) +  \dd(\Phi_{n,\omega}^{(\beta)}\cdot Y, \Phi_{n,\omega}^{(\alpha)}\cdot Y)   \\
    & \le 2 |\beta-\alpha|\sum_{m=0}^{n-1} F_{\theta^m\omega} + c(\Phi^{(\alpha)}_{n,\omega}) \
    .
    \end{split}
    \end{equation}
    It follows that 
    \begin{equation}
        c(\Phi_n^{(\beta)}) \leq 2 |\beta-\alpha|\sum_{m=0}^{n-1} F_{\theta^m\omega} +  c(\Phi^{(\alpha)}_n).
    \end{equation}
    The bound~\eqref{eq:equiconlemma} follows from this inequality together with that for $\alpha$ and $\beta$ interchanged.
\end{proof}

\begin{lemma}
\label{lem:simul-alpha-ii}
    If {\textnormal{\ref{A2}}} and~{\textnormal{\ref{A3}}} hold, then  there is a $\theta$-invariant set $\Omega_0'$ of full $\pp$-measure with the following property:
    $Z^{(\alpha)}_\omega$ and $\Zp^{(\alpha)}_\omega$ belong to $\mathbb{S}^\circ$ and satisfy the identities~\eqref{eq:cocycle}
    and~\eqref{eq:fundbound} 
    for all $\omega\in\Omega_0$, $\alpha\in\rr$, $Y\in \mathbb{S}$ and $N\in\nn$.
\end{lemma}

\begin{proof}
    For fixed $\omega \in \Omega_0$ as in Lemmas~\ref{lem:simul-pos-def} and~\ref{lem:simul-alpha} and $\alpha \in \rr$, the proof of Lemma~3.12 in~\cite{MS22} can be followed line by line to deduce that $Z^{(\alpha)}_\omega$ and $\Zp^{(\alpha)}_\omega$ belong to $\mathbb{S}^\circ$ and satisfy~\eqref{eq:cocycle}
    and~\eqref{eq:fundbound} for all~$Y \in \mathbb{S}$ and~$N\in\nn$. Indeed, the $\pp$-almost sure properties that are used in that proof are those concerned by Lemmas~\ref{lem:simul-pos-def} and~\ref{lem:simul-alpha}.
    Finally, by $\theta$-invariance of $\pp$ and countable additivity, taking
    \[
        \Omega'_0 := \bigcap_{n\in\nn} \theta^{1-n}(\Omega_0),
    \]
    gives a $\theta$-invariant set with the same properties.
\end{proof}

\subsection{Regularity of the limiting objects}

Recall that, because we want to employ the G\"artner--Ellis theorem, we wish to show differentiability of~$\alpha \mapsto \lambda(\alpha)$.
This will be done in several steps. 

\begin{proposition}
\label{prop:continuous-Z}
    Suppose that {\textnormal{\ref{A2}}} and~{\textnormal{\ref{A3}}} hold. Then, for $\pp$-almost all~$\omega$, the function~$\alpha \mapsto Z_\omega^{(\alpha)}$ is continuous on~$\rr$ (in either of the topologies involved in Lemma~\ref{lem:d-to-tr}).
\end{proposition}

\begin{proof}
    We will show that $\alpha \mapsto Z^{(\alpha)}_\omega$ is continuous on the set of full $\pp$-measure where the results of Section~\ref{ssec:simult} hold. Throughout the proof, let $\omega$ denote a fixed element of this set. To lighten the notation, we suppress $\omega$ from formulas by replacing subscripts of the form~``$_{\theta^{j-1}\omega}$'' or ``$_{j,\omega}$'' with simply ``$_j$''.
    
    Fix $\alpha \in \rr$ and let $\epsilon >0$.  By Lemmas~\ref{lem:d-to-tr} and~\ref{lem:simul-alpha-ii},
    \[ \dd(Z_0^{(\alpha)},Z_0^{(\beta)}) \leq \frac{1}{\delta} \left \| Z_0^{(\alpha)} - Z_0^{(\beta)} \right \|_1 = \frac{1}{\delta} \left \| \Phi_{-N}^{(\alpha)} \cdot Z_{-N}^{(\alpha)}  - \Phi_{-N}^{(\beta)} \cdot Z_{-N}^{(\beta)}  \right \|_1 , \]
    where  $\delta= \min \sigma(Z_0^{(\alpha)}) >0$ and $N\in \nn$.
    Using the triangle inequality, we find that for all $N\in \nn$ and $\beta \in \rr$,
    \begin{align*}
        \dd(Z_0^{(\alpha)},Z_0^{(\beta)}) &\le \frac{1}{\delta} \left \| \Phi_{-N}^{(\alpha)} \cdot Z_{-N}^{(\alpha)} - \Phi_{-N}^{(\alpha)} \cdot Z_{-N}^{(\beta)} \right \|_1 + \frac{1}{\delta}  \left \| \Phi_{-N}^{(\alpha)} \cdot Z_{-N}^{(\beta)} - \Phi_{-N}^{(\beta)} \cdot Z_{-N}^{(\beta)} \right \|_1 \\
        &\le \frac{2}{\delta} \dd (\Phi_{-N}^{(\alpha)} \cdot Z_{-N}^{(\alpha)}, \Phi_{-N}^{(\alpha)} \cdot Z_{-N}^{(\beta)})  + \frac{1}{\delta}  \left \| \Phi_{-N}^{(\alpha)} \cdot Z_{-N}^{(\beta)} - \Phi_{-N}^{(\beta)} \cdot Z_{-N}^{(\beta)} \right \|_1 \\
        &\le \frac{2}{\delta} c(\Phi_{-N}^{(\alpha)}) + \frac{1}{\delta}  \left \| \Phi_{-N}^{(\alpha)} \cdot Z_{-N}^{(\beta)} - \Phi_{-N}^{(\beta)} \cdot Z_{-N}^{(\beta)} \right \|_1 .
    \end{align*}
    where we have used Lemma~\ref{lem:d-to-tr} again in the second line.  By Lemma~\ref{lem:simul-alpha}, there exists~$N$ such that~$c(\Phi_{-N}^{(\alpha)}) < \tfrac{\delta}{2} \epsilon$, so that 
    \[ \dd(Z_0^{(\alpha)},Z_0^{(\beta)}) < \epsilon +  \frac{1}{\delta} \left \| \Phi_{-N}^{(\alpha)} \cdot Z_{-N}^{(\beta)} - \Phi_{-N}^{(\beta)} \cdot Z_{-N}^{(\beta)} \right \|_1 .  \]
    To complete the proof it suffices to show that
    \begin{equation}\label{eq:completecontinuity}
    \lim_{\beta \to \alpha} \left \| \Phi_{-N}^{(\alpha)} \cdot Z_{-N}^{(\beta)} - \Phi_{-N}^{(\beta)} \cdot Z_{-N}^{(\beta)} \right \|_1 = 0,
    \end{equation}
    as then $\limsup_{\beta\to \alpha} \dd(Z_0^{(\alpha)},Z_0^{(\beta)}) \leq \epsilon$ for every $\epsilon$.

    To prove \eqref{eq:completecontinuity}, note that
    \[ \left\| \Phi_{-N}^{(\alpha)} \cdot Z_{-N}^{(\beta)} - \Phi_{-N}^{(\beta)} \cdot Z_{-N}^{(\beta)} \right\|_1 \leq \sup \left\{\left\|\Phi_{-N}^{(\alpha)} \cdot Z - \Phi_{-N}^{(\beta)} \cdot Z \right\|_1 \ : \ Z \in \mathbb{S} \right\} . \]
    Using \eqref{eq:varphialpha} we find, for each $N$, that 
    \[  \Phi_{-N}^{(\beta)} = \sum_{j=1}^{M} \Exp{- F_j} \Psi_{j} , \]
    where $\sum_{j=1}^{M} \Psi_j = \Phi_{-N}^{(0)}$,  $M = \cal{A}^N$, $\Psi_j$ are CP maps, and $F_j$ are certain random functions (which are finite for the fixed $\omega$ under consideration).
    It follows that 
    \[ \Exp{-C |\beta-\alpha|} \Phi_{-N}^{(\alpha)}\cdot Z \leq \Phi_{-N}^{(\beta)}\cdot Z \leq \Exp{C|\beta-\alpha|} \Phi_{-N}^{(\alpha)}\cdot Z  \]
    with $C= 2 \max_j |F_j|$. 
    Thus 
    \[ \left \| \Phi_{-N}^{(\alpha)} \cdot Z_{-N}^{(\beta)} - \Phi_{-N}^{(\beta)} \cdot Z_{-N}^{(\beta)} \right \|_1 \leq  \sup_{Z\in \mathbb{S}} \left \| \Phi_{-N}^{(\alpha)} \cdot Z \right \|_1 \left ( \Exp{C|\beta-\alpha|} - 1 \right ) , \]
    from which \eqref{eq:completecontinuity} follows.
\end{proof}

\begin{lemma}
\label{lem:Delta-lambda}
    If \ref{A1},~\ref{A2} and~\ref{A3} hold, then
    \begin{align*}
        \lambda(\alpha)-\lambda(\beta) 
            &= \int_\Omega \log 
            \frac{\tr [\Zp_{\theta\omega}^{(\alpha)} \varphi_{\omega}^{(\alpha)} (Z_{\theta^{-1}\omega}^{(\beta)})]}
            {\tr [\Zp_{\theta\omega}^{(\alpha)} \varphi_{\omega}^{(\beta)} (Z_{\theta^{-1}\omega}^{(\beta)})]}
            \dd\pp(\omega)
    \end{align*}
    for all~$\alpha$ and $\beta$ in~$\rr$.
\end{lemma}

\begin{proof}
    We consider only $\omega$ in a set of full $\pp$-measure where the results of Section~\ref{ssec:simult} hold. Considering the difference between the integrands for~\eqref{eq:Lyap-as-average} in $\alpha$ and in $\beta$, and then adding and subtracting some intermediate terms, we find that
    \[ \lambda(\alpha)-\lambda(\beta)   = \int_\Omega \Delta_\omega(\alpha,\beta) \dd \pp(\omega) ,  \]
    with
    \begin{align*}
        \Delta_\omega(\alpha,\beta)
         &:=\log  \tr [\Zp_{\theta\omega}^{(\alpha)} \varphi_{\omega}^{(\alpha)} (Z_{\theta^{-1}\omega}^{(\alpha)})] 
        - \log  \tr [\Zp_{\theta\omega}^{(\beta)} \varphi_{\omega}^{(\beta)} (Z_{\theta^{-1}\omega}^{(\beta)})]
        - \log \tr [\Zp_{\theta\omega}^{(\alpha)} Z_{\omega}^{(\alpha)}] 
        + \log \tr [\Zp_{\theta\omega}^{(\beta)} Z_{\omega}^{(\beta)}] 
        \\ 
        &\phantom{:}= 
            \log  \tr [\Zp_{\theta\omega}^{(\alpha)} \varphi_{\omega}^{(\alpha)} (Z_{\theta^{-1}\omega}^{(\alpha)})] 
            - \log  \tr [\Zp_{\theta\omega}^{(\alpha)} \varphi_{\omega}^{(\alpha)} (Z_{\theta^{-1}\omega}^{(\beta)})] \\
            &\qquad\qquad + \log  \tr [\Zp_{\theta\omega}^{(\alpha)} \varphi_{\omega}^{(\alpha)} (Z_{\theta^{-1}\omega}^{(\beta)})]
            - \log  \tr [\Zp_{\theta\omega}^{(\alpha)} \varphi_{\omega}^{(\beta)} (Z_{\theta^{-1}\omega}^{(\beta)})] \\
            &\qquad\qquad + \log  \tr [\Zp_{\theta\omega}^{(\alpha)} \varphi_{\omega}^{(\beta)} (Z_{\theta^{-1}\omega}^{(\beta)})]
            - \log  \tr [\Zp_{\theta\omega}^{(\beta)} \varphi_{\omega}^{(\beta)} (Z_{\theta^{-1}\omega}^{(\beta)})] \\
            &\qquad\qquad - \log \tr [\Zp_{\theta\omega}^{(\alpha)} Z_{\omega}^{(\alpha)}]
            + \log \tr [\Zp_{\theta\omega}^{(\alpha)} Z_{\omega}^{(\beta)}] \\
            &\qquad\qquad - \log \tr [\Zp_{\theta\omega}^{(\alpha)} Z_{\omega}^{(\beta)}]
            + \log \tr [\Zp_{\theta\omega}^{(\beta)} Z_{\omega}^{(\beta)}].
    \end{align*}
    By~\eqref{eq:cocycle}, the third and fifth lines in the last right-hand side cancel, so 
    \begin{align*}
        \Delta_\omega(\alpha,\beta)
        &= 
            \log  \tr [\Zp_{\theta\omega}^{(\alpha)} \varphi_{\omega}^{(\alpha)} (Z_{\theta^{-1}\omega}^{(\alpha)})] 
            - \log  \tr [\Zp_{\theta\omega}^{(\alpha)} \varphi_{\omega}^{(\alpha)} (Z_{\theta^{-1}\omega}^{(\beta)})] \\
            &\qquad\qquad + \log  \tr [\Zp_{\theta\omega}^{(\alpha)} \varphi_{\omega}^{(\alpha)} (Z_{\theta^{-1}\omega}^{(\beta)})]
            - \log  \tr [\Zp_{\theta\omega}^{(\alpha)} \varphi_{\omega}^{(\beta)} (Z_{\theta^{-1}\omega}^{(\beta)})] \\
            &\qquad\qquad - \log \tr [\Zp_{\theta\omega}^{(\alpha)} Z_{\omega}^{(\alpha)}]
            + \log \tr [\Zp_{\theta\omega}^{(\alpha)} Z_{\omega}^{(\beta)}] \\
        &= 
            \log  \tr [\Zp_{\omega}^{(\alpha)}  Z_{\theta^{-1}\omega}^{(\alpha)}] 
            - \log  \tr [\Zp_{\omega}^{(\alpha)}  Z_{\theta^{-1}\omega}^{(\beta)}] \\
            &\qquad\qquad + \log  \tr [\Zp_{\theta\omega}^{(\alpha)} \varphi_{\omega}^{(\alpha)} (Z_{\theta^{-1}\omega}^{(\beta)})]
            - \log  \tr [\Zp_{\theta\omega}^{(\alpha)} \varphi_{\omega}^{(\beta)} (Z_{\theta^{-1}\omega}^{(\beta)})] \\
            &\qquad\qquad - \log \tr [\Zp_{\theta\omega}^{(\alpha)} Z_{\omega}^{(\alpha)}]
            + \log \tr [\Zp_{\theta\omega}^{(\alpha)} Z_{\omega}^{(\beta)}],
    \end{align*}
    where we have used \eqref{eq:cocycle} again to obtain the final expression.
    Note that this is a relation of the form
    \begin{align*}
        \Delta_\omega(\alpha,\beta)
            &= \log 
            \frac{\tr [\Zp_{\theta\omega}^{(\alpha)} \varphi_{\omega}^{(\alpha)} (Z_{\theta^{-1}\omega}^{(\beta)})]}
            {\tr [\Zp_{\theta\omega}^{(\alpha)} \varphi_{\omega}^{(\beta)} (Z_{\theta^{-1}\omega}^{(\beta)})]} + [g_\omega - g_{\theta\omega}].
    \end{align*}
    Recall that $\Delta_\omega(\alpha,\beta)$ on the left-hand side is a difference of two integrable functions. On the right-hand side, the logarithm admits the straightforward bound
    \begin{equation}\label{eq:independentofalphabeta}
        \left|\log 
            \frac{\tr [\Zp_{\theta\omega}^{(\alpha)} \varphi_{\omega}^{(\alpha)} (Z_{\theta^{-1}\omega}^{(\beta)})]}
            {\tr [\Zp_{\theta\omega}^{(\alpha)} \varphi_{\omega}^{(\beta)} (Z_{\theta^{-1}\omega}^{(\beta)})]}\right|
            \leq 2F_\omega|\beta-\alpha|
    \end{equation}
    ---\,thanks to~\ref{A3}\,---\,and is therefore also integrable.
    This implies that $g-g\circ\theta$ is integrable as well.
    It thus follows from a consequence of Birkhoff's theorem (stated as Lemma~\ref{lem:shift-cancellation} below) that
    \[
         \lambda(\alpha) - \lambda(\beta) = 
        \int_\Omega \Delta_\omega(\alpha,\beta) \dd\pp(\omega)
        = \int_\Omega \log 
            \frac{\tr [\Zp_{\theta\omega}^{(\alpha)} \varphi_{\omega}^{(\alpha)} (Z_{\theta^{-1}\omega}^{(\beta)})]}
            {\tr [\Zp_{\theta\omega}^{(\alpha)} \varphi_{\omega}^{(\beta)} (Z_{\theta^{-1}\omega}^{(\beta)})]}
            \dd\pp(\omega).  \qedhere
    \]
\end{proof}

\begin{proposition}
\label{prop:lambda-diff}
    If~\ref{A1},~\ref{A2} and~\ref{A3} hold, then
    the function~$\alpha \mapsto \lambda(\alpha)$ is differentiable on~$\rr$.
\end{proposition}

\begin{proof}
    Again, we consider only $\omega$ in a set of full $\pp$-measure where the results of Section~\ref{ssec:simult} hold.
    To lighten the notation, we will replace subscripts of the form~``$_{\theta^j\omega}$'' with simply ``$_j$''. 
    By Lemma~\ref{lem:Delta-lambda} (and its proof), we have 
    \begin{equation}
        \frac{\lambda(\beta)-\lambda(\alpha)}{\beta-\alpha}
            = \int_\Omega \frac{1}{\beta-\alpha} \log \frac{ \tr [\Zp_1^{(\alpha)} \varphi_0^{(\beta)} (Z_{-1}^{(\beta)})]}{\tr [\Zp_1^{(\alpha)} \varphi_0^{(\alpha)} (Z_{-1}^{(\beta)})]} \dd\pp
    \end{equation}
    with the integrand on the right-hand side being dominated by an integrable function of $\omega$ independent of~$\alpha$ and $\beta$ \textemdash \ see \eqref{eq:independentofalphabeta}.
    The Lebesgue dominated convergence theorem thus guarantees that we can exchange the integral with respect to~$\pp$ and the limit $\beta \to \alpha$. Therefore, the proof of differentiability reduces to showing that the limit
    \[ 
        \lim_{\beta\to\alpha} \frac{1}{\beta-\alpha} \log \frac{ \tr [\Zp_1^{(\alpha)} \varphi_0^{(\beta)} (Z_{-1}^{(\beta)})]}{\tr [\Zp_1^{(\alpha)} \varphi_0^{(\alpha)} (Z_{-1}^{(\beta)})]}
    \]
    exists $\pp$-almost surely.

    Note as an intermediate step that, for all~$\beta$ and $\gamma$, we have
    \begin{align*}
    \label{eq:proto-deriv-lambda}
        \frac{\dd}{\dd\gamma} \log\tr\left[\Zp_1^{(\alpha)} \varphi_0^{(\gamma)} (Z_{-1}^{(\beta)})\right]
        &= \frac{\tr\left[\Zp_1^{(\alpha)} \frac{\dd}{\dd\gamma}\varphi_0^{(\gamma)} (Z_{-1}^{(\beta)})\right]}{\tr\left[\Zp_1^{(\alpha)} \varphi_0^{(\gamma)} (Z_{-1}^{(\beta)})\right]},
    \end{align*}
    which is continuous in $\gamma$ and in $\beta$, thanks to the precise form of the deformations and to Proposition~\ref{prop:continuous-Z}.
    Hence, by the mean value theorem, for every $\beta$, there exists $\gamma_0$ between $\alpha$ and $\beta$ such that
    \begin{align*}
        \frac{1}{\beta-\alpha} \log \frac{ \tr [\Zp_1^{(\alpha)} \varphi_0^{(\beta)} (Z_{-1}^{(\beta)})]}{\tr [\Zp_1^{(\alpha)} \varphi_0^{(\alpha)} (Z_{-1}^{(\beta)})]}
        &= \frac
        {\tr
        \left[\Zp_1^{(\alpha)} \left.\frac{\dd}{\dd\gamma} \varphi_0^{(\gamma)}\right|_{\gamma=\gamma_0} (Z_{-1}^{(\beta)}) \right]}
        {\tr\left[\Zp_1^{(\alpha)} \varphi_0^{(\gamma_0)} (Z_{-1}^{(\beta)})\right]}.
    \end{align*}
    Taking $\beta \to \alpha$ makes $\gamma_0 \to \alpha$ and we find
    \begin{equation*}
        \lim_{\beta \to \alpha}\frac{1}{\beta-\alpha} \log \frac{ \tr [\Zp_1^{(\alpha)} \varphi_0^{(\beta)} (Z_{-1}^{(\beta)})]}{\tr [\Zp_1^{(\alpha)} \varphi_0^{(\alpha)} (Z_{-1}^{(\beta)})]}
        = \frac{\tr\left[\Zp_1^{(\alpha)} \left.\frac{\dd}{\dd\gamma}\varphi_0^{(\gamma)}\right|_{\gamma=\alpha} (Z_{-1}^{(\alpha)})\right]}{\tr\left[\Zp_1^{(\alpha)} \varphi_0^{(\alpha)} (Z_{-1}^{(\alpha)})\right]}
    \end{equation*}
    by Proposition~\ref{prop:continuous-Z}.
\end{proof}

\section{Proofs for Section~\ref{sec:EP}}
\label{sec:EP-proofs}

We start with our statement concerning symmetry in~$\alpha$ about $\alpha = \tfrac 12$.

\begin{proof}[Proof of Lemma~\ref{lem:dual-1-minus}]
    With $\alpha$ and $\omega$ fixed but arbitrary, it suffices to show that 
    $$\tr(\tau X\tau \varphi_{\mathsf{T}\omega}^{(\alpha)} (Y)) = \tr(\tau Y \tau \varphi_{\mathsf{T}\omega}^{(1-\alpha)} (X))
    $$
    for every $X$ and $Y$ in $\cB(\cc^d)$. Following closely computations found in e.g.~\cite[\S{3.2}]{HJPR18} or~\cite[\S\S{5.3--5.4}]{BJP22}, we indeed find
    \begin{align*}
        \tr(\tau X\tau \varphi_{\mathsf{T}\omega}^{(\alpha)} (Y)) 
            &= \sum_{\epsilon,\epsilon'} \Exp{-\alpha \beta_{\mathsf{T}\omega}(E_{\mathsf{T}\omega,\epsilon'} - E_{\mathsf{T}\omega,\epsilon})} \tr[(\tau X \tau \otimes \one ) (\one \otimes \pi_{\mathsf{T}\omega,\epsilon'} ) U_{\mathsf{T}\omega} (Y \otimes \xi_{\mathsf{T}\omega}\pi_{\mathsf{T}\omega,\epsilon}) U^*_{\mathsf{T}\omega}] \\
            &= \sum_{\epsilon',\epsilon}\Exp{-(1-\alpha) \beta_{\mathsf{T}\omega}(E_{\mathsf{T}\omega,\epsilon'} - E_{\mathsf{T}\omega,\epsilon})} \tr[(\tau X \tau \otimes \xi_{\mathsf{T}\omega} \pi_{\mathsf{T}\omega,\epsilon}) U_{\mathsf{T}\omega} (Y \otimes \one)(\one \otimes \pi_{\mathsf{T}\omega,\epsilon'}) U^*_{\mathsf{T}\omega}] \\
            &= \sum_{\epsilon',\epsilon}\Exp{-(1-\alpha) \beta_{\mathsf{T}\omega}(E_{\mathsf{T}\omega,\epsilon'} - E_{\mathsf{T}\omega,\epsilon})} \tr[( X \otimes \tau' \xi_{\mathsf{T}\omega} \pi_{\mathsf{T}\omega,\epsilon}) U_{\omega}^* (\tau Y \tau \otimes \one)(\one \otimes \pi_{\mathsf{T}\omega,\epsilon'}\tau') U_{\omega}] \\
            &= \sum_{\epsilon',\epsilon}\Exp{-(1-\alpha) \beta_{\mathsf{T}\omega}(E_{\mathsf{T}\omega,\epsilon'} - E_{\mathsf{T}\omega,\epsilon})} \tr[( X \otimes \xi_{\omega} \pi_{\omega,\epsilon}) U_{\omega}^* (\tau Y \tau \otimes \one)(\one \otimes \pi_{\omega,\epsilon'}) U_{\omega}] \\
            &= \tr(\tau Y \tau \varphi_{\omega}^{(1-\alpha)} (X)),
    \end{align*}
    as desired. Note that we have used Remark~\ref{rem:xi-labeling} in the perhaps more subtle step where $\tau' \xi_{\mathsf{T}\omega} \pi_{\mathsf{T}\omega,\epsilon}$ becomes $ \xi_{\omega} \pi_{\omega,\epsilon} \tau'$ and $\pi_{\mathsf{T}\omega,\epsilon'}\tau'$ becomes $\tau' \pi_{\omega,\epsilon'}$.
\end{proof}

We now turn to bounding $\log p_{\omega,\rho}^{(n)} - \log p_{\mathsf{T}\theta^{n-1}\omega,\rho}^{(n)}$ in the information-theoretic notion of entropy production in terms of $\Sigma_{\omega,n}$ in the Clausius-like notion of entropy production.

\begin{proof}[Proof of Lemma~\ref{lem:first-EP-bound}]
    A computation similar to that in the proof of Lemma~\ref{lem:dual-1-minus} shows that
    \begin{align*}
        \psi_{\mathsf{T}\omega,(\epsilon,\epsilon')}(\tau \rho \tau) 
        &=  \Exp{\beta_\omega(E_{\omega,\epsilon'}-E_{\omega,\epsilon})} \tau \psi_{\omega,(\epsilon',\epsilon)}^*(\rho)\tau.
    \end{align*}
    Now, this implies
    \begin{align*}
        p^{(n)}_{\omega,\rho}((\epsilon_1,\epsilon'_1), \dotsc, (\epsilon_n,\epsilon'_n))
            &= \tr[(\psi_{\theta^{n-1}\omega,(\epsilon_n,\epsilon'_n)}\circ\dotsb\circ\psi_{\omega,(\epsilon_1,\epsilon'_1))})(\rho)] \\
            &= \Exp{\Sigma_{\omega,n} }\tr[(\psi^*_{\mathsf{T}\theta^{n-1}\omega,(\epsilon_n,\epsilon'_n)}\circ\dotsb\circ\psi^*_{\mathsf{T}\omega,(\epsilon_1,\epsilon'_1)})(\tau\rho\tau)] \\
            &= \Exp{\Sigma_{\omega,n} }\tr[\tau\rho\tau \ (\psi_{\mathsf{T}\omega,(\epsilon_1,\epsilon'_1)}\circ\dotsb\circ\psi_{\mathsf{T}\theta^{n-1}\omega,(\epsilon_n,\epsilon'_n)})(\one)] \\
            &= \Exp{\Sigma_{\omega,n} } \tr[\tau\rho\tau \ (\psi_{\theta^{n-1}\mathsf{T}\theta^{n-1}\omega,(\epsilon_1,\epsilon'_1)}\circ\dotsb\circ\psi_{\mathsf{T}\theta^{n-1}\omega,(\epsilon_n,\epsilon'_n)})(\one)].
    \end{align*}
    It then indeed follows from the definition of~$\Delta_\rho$ and the fact that $\sp(\tau\rho\tau) = \sp \rho$ that 
    \begin{align*}
      \Exp{\Sigma_{\omega,n}} \Delta_\rho^{-1} {p}^{(n)}_{\mathsf{T}\theta^{n-1}\omega,\rho} \leq p^{(n)}_{\omega,\rho}
            &\leq \Exp{\Sigma_{\omega,n}} \Delta_\rho {p}^{(n)}_{\mathsf{T}\theta^{n-1}\omega,\rho}.
    \end{align*}
    Again, we refer the reader to~\cite[App.\,C]{HJPR18} and~\cite[\S{5.4}]{BJP22} for similar computations.
\end{proof}

We can then relate the large deviations of the information-theoretic notion of entropy production to the Clausius-like notion of entropy production using standard techniques.

\begin{proof}[Proof of Theorem~\ref{thm:main-EP}]
    Lemma~\ref{lem:first-EP-bound} guarantees that, for every $\rho \in \mathbb{S}^\circ$, the sequences~$(\sigma_{\omega,\rho,n})_{n\in\nn}$ and $(\Sigma_{\omega,n})_{n\in\nn}$ are exponentially equivalent and thus have equivalent large deviations; see e.g.~\cite[\S{4.2.2}]{DeZe}. Therefore, the large deviation principle follows from Theorem~\ref{thm:main-abstract}.
    
    Lemma~\ref{lem:dual-1-minus} implies that $Z_{\mathsf{T}\omega}^{(\alpha)} = \tau \Zp_{\omega}^{(1-\alpha)} \tau $ for all~$\alpha \in \rr$ and $\omega \in \Omega$. Hence, by~\eqref{eq:Lyap-as-average}, TRI and again Lemma~\ref{lem:dual-1-minus}, we have 
    \begin{align*}
        \lambda(\alpha) 
            &= \int \log  \tr [\Zp_{\theta\omega}^{(\alpha)} \varphi_{\omega}^{(\alpha)} (Z_{\theta^{-1}\omega}^{(\alpha)})] - \log \tr [\Zp_{\theta\omega}^{(\alpha)} Z_{\omega}^{(\alpha)}] \dd\pp(\omega) \\
            &= \int \log  \tr [\Zp_{\theta\mathsf{T}\omega}^{(\alpha)} \varphi_{\mathsf{T}\omega}^{(\alpha)} (Z_{\theta^{-1}\mathsf{T}\omega}^{(\alpha)})] - \log \tr [\Zp_{\theta\mathsf{T}\omega}^{(\alpha)} Z_{\mathsf{T}\omega}^{(\alpha)}] \dd\pp(\omega) \\
            &= \int \log  \tr [\Zp_{\mathsf{T}\theta^{-1}\omega}^{(\alpha)} \varphi_{\mathsf{T}\omega}^{(\alpha)} (Z_{\mathsf{T}\theta\omega}^{(\alpha)})] - \log \tr [\Zp_{\mathsf{T}\theta^{-1}\omega}^{(\alpha)} Z_{\mathsf{T}\omega}^{(\alpha)}] \dd\pp(\omega) \\
            &= \int \log  \tr [\tau Z_{\theta^{-1}\omega}^{(1-\alpha)} \tau \tau (\varphi_{\omega}^{(1-\alpha)})^*( \Zp_{\theta\omega}^{(1-\alpha)}) \tau] - \log \tr [\tau Z_{\theta^{-1}\omega}^{(1-\alpha)} \tau\tau \Zp_{\omega}^{(1-\alpha)}\tau] \dd\pp(\omega) \\
            &= \int \log  \tr [ \Zp_{\theta\omega}^{(1-\alpha)} \varphi_{\omega}^{(1-\alpha)}( Z_{\theta^{-1}\omega}^{(1-\alpha)})] - \log \tr [ \Zp_{\omega}^{(1-\alpha)} Z_{\theta^{-1}\omega}^{(1-\alpha)} ] \dd\pp(\omega) \\
            &= \lambda(1-\alpha)
    \end{align*}
    for all~$\alpha \in \rr$. Hence, the Gallavotti--Cohen symmetry follows by Legendre duality.
\end{proof}

\appendix
\section{Abstract measure-theoretic results}
\label{app:together}

Throughout this appendix, $(\Omega, \theta, \pp)$ is an ergodic dynamical system. We do not claim novelty of the results below, but  have included explicit statements and proofs for the convenience of the reader.
\subsection{Almost sure convergence of convex functions}
\begin{lemma}\label{lem:convexlemma} 
    Let $I\subset \rr$ be an open interval, and for each $N\in \nn$, let $\omega \mapsto G_{N;\omega}(\cdot)$ be a random convex function  with domain $I$.  Suppose, for each $\alpha\in I$, that 
    $$\lim_{N\to \infty} G_{N}(\alpha) = \lambda(\alpha) \ \quad \text{almost surely,}$$ 
    where $\lambda$ is a non-random convex function on $I$. Then with probability one 
    $$
        \lim_{N\to \infty} G_{N}(\alpha) = \lambda(\alpha) \quad \text{ for all }\alpha \in I.
    $$
\end{lemma}

\begin{remark}
    This particular result does not involve the transformation $\theta$; it holds on a general probability space $(\Omega,\pp)$.
\end{remark}

\begin{proof}
    By countable additivity, there exists a set $\Omega_0$ of full measure such that for $\omega\in \Omega_0$
    $$
        \lambda(q) = \lim_{N\to \infty} G_{N;\omega}(q) \qquad  \text{for all }q\in\qq \cap I .
    $$
    Fix a point $\omega\in \Omega_0$ and let $\alpha \in I$ be arbitrary.  If $q<\alpha <r$ are rational numbers, then
    $$
        \limsup_{N\to \infty} G_{N;\omega}(\alpha) \leq \limsup_{N\to \infty} \left (\frac{r-\alpha}{r-q} G_{N;\omega}(q) + \frac{\alpha-q}{r-q} G_{N;\omega}(r) \right ) = \frac{r-\alpha}{r-q} \lambda(q) + \frac{\alpha-q}{r-q} \lambda(r) . 
    $$
    Since $\lambda$ is continuous (being convex), we conclude by taking $q,r\to \alpha$ that
    \begin{equation}
        \limsup_{N\to\infty} G_{N;\omega}(\alpha) \leq \lambda(\alpha) . \label{eq:limsuplambda}
    \end{equation}
    Similarly, if $q<r<\alpha$ are rational, we have
    \begin{equation} 
    \begin{split}
        \lambda(r) &= \lim_{N\to \infty} G_{N;\omega}(r) \\ 
        &\le \liminf_{N\to \infty} \left ( \frac{\alpha- r}{\alpha-q} G_{N;\omega}(q) +  \frac{r-q}{\alpha-q} G_{N;\omega}(\alpha) \right ) \\ 
        &= \frac{\alpha- r}{\alpha-q} \lambda(q) + \frac{r-q}{\alpha-q} \liminf_{N\to \infty} G_{N;\omega}(\alpha) .
    \end{split}
    \end{equation}
    Taking $r \to \alpha$ with $q$ fixed, by continuity of $\lambda$ again, we conclude that
    \begin{equation} 
    \label{eq:liminflambda}
        \lambda(\alpha) \leq \liminf_{N\to \infty} G_{N;\omega}(\alpha) .
    \end{equation}
    The result follows from \eqref{eq:limsuplambda} and \eqref{eq:liminflambda}.
\end{proof}

\subsection{A sub-additive ergodic upper bound for continuous families}

Let $c_{n}^{(\alpha)}:\Omega \to (0,1] $ be a family of $(0,1]$-valued random variables indexed by $n\in \nn$ and $\alpha \in I\subset \rr$ with $I$ an interval. Suppose that 
\begin{enumerate}
    \item 
        For each fixed $\alpha$, the sequence is almost surely submultiplicative in the sense that, for $\pp$-almost every~$\omega$,
        \begin{equation}
        \label{eq:appendixsubadd} 
            c_{n}^{(\alpha)}(\omega) \leq c_{n-m}^{(\alpha)}(\theta^m\omega) c_{m}^{(\alpha)}(\omega) 
        \end{equation}
        whenever $n > m$.
    \item 
        For every~$n$, there exists a function $\epsilon^{(n)} : \Omega \times (0,1) \to (0,\infty)$ with the following properties: 
        \begin{itemize}
        \item 
            for $\pp$-almost every~$\omega$, the function $\epsilon^{(n)}(\omega,\,\cdot\,)$ is decreasing and
            \begin{equation}
            \label{eq:appendixequicon}
                \left|c_{n}^{(\alpha)}(\omega) - c_{n}^{(\beta)}(\omega)\right| < \epsilon^{(n)}(\omega,\delta)
            \end{equation}
            whenever $|\beta-\alpha| < \delta$, and
            \item the function $\epsilon^{(n)}(\,\cdot\,,\delta)$ is in $L^1(\Omega,\dd\pp)$ for all $\delta$ and
            \begin{equation}
            \label{eq:mod-int-to-0}
                \lim_{\delta\to 0}\int \epsilon^{(n)}(\omega,\delta) \dd\pp(\omega) = 0.
            \end{equation}
        \end{itemize}
\end{enumerate}
\begin{lemma}\label{lem:appendixlemma}
    Under the above Assumptions~1 and 2, there is set $\Omega_0$  of a full $\pp$-measure such that
    \begin{equation}\label{eq:appendixlemma}
        \limsup_{n\to \infty} \frac{1}{n} \log c_{n}^{(\alpha)}(\omega) \leq \inf_{n\in\nn} \frac{1}{n} \int \log c_{n}^{(\alpha)} \dd\pp
    \end{equation} 
    for all~$\omega\in \Omega_0$ and~$\alpha \in I$.
\end{lemma}

\begin{remark}
    Since $(\int \log c_{n}^{(\alpha)} \dd\pp )_{n\in \nn} $ is a subadditive sequence, it follows from Fekete's lemma that
    \begin{equation}
    \label{eq:Fekete}  
        \inf_{n\in\nn} \frac{1}{n} \int \log c_{n}^{(\alpha)} \dd\pp  = \lim_{n\to \infty} \frac{1}{n} \int \log c_{n}^{(\alpha)} \dd\pp .
    \end{equation}
\end{remark}

\begin{proof}
    We begin with two observations. First, by submultiplicativity, 
    \begin{equation}\label{eq:firstsubad}
        \log c_{n}^{(\alpha)} \leq  \sum_{m=0}^{\lfloor \frac{n}{k}\rfloor-1} \log c_{k}^{(\alpha)}\circ \theta^{mk+j}
    \end{equation}
    for every choice of $0 \leq j < k < n$.
    Averaging over $j\in\{0,1,\ldots,k-1\}$ then yields the bound
    \begin{equation}\label{eq:secondsubad}
    \log c_{n}^{(\alpha)} \leq \frac{1}{k} \sum_{m=0}^{k\lfloor\frac{n}{k} \rfloor -1} \log c_{k}^{(\alpha)}\circ \theta^{m}.
    \end{equation}
    Second, by Birkhoff's theorem and countable additivity, there exists a set $\Omega_0$ of full $\pp$-measure such that, for all $\omega\in \Omega_0$, $k \in \nn$, $\beta\in\qq$, $t \in (0,1) \cap \qq$ and $\delta \in (0,1) \cap \qq$, we have
    \begin{align}
        \lim_{n\to \infty} \frac{1}{n} \sum_{m=0}^{n-1} \log \left ( \max\left\{t, c_{k}^{(\beta)}(\theta^{m}\omega)\right\} \right ) &= \int \log \left (\max\left\{t, c_{k}^{(\beta)}\right\} \right ) \dd\pp
    \label{eq:birkhoff-app-1} \\
    \intertext{and}
        \lim_{n\to \infty} \frac{1}{n} \sum_{m=0}^{n-1} \epsilon^{(k)}(\theta^m\omega,\delta) &= \int \epsilon^{(k)}(\,\cdot\,,\delta) \dd\pp .
    \label{eq:birkhoff-app-2}
    \end{align}

    Now fix $\omega\in \Omega_0$ and $\alpha\in I$. Let also $k \in \nn$, $t \in (0,1] \cap \qq$ and $\delta \in (0,1] \cap \qq$ be arbitrary. By~\eqref{eq:secondsubad} and~\eqref{eq:appendixequicon}, we have
    \begin{equation} 
    \begin{split}
        \log c_{n}^{(\alpha)}(\omega) &\le \frac{1}{k} \sum_{m=0}^{k\lfloor\frac{n}{k} \rfloor -1} \log \left ( \max\left\{t, c_{k}^{(\alpha)}(\theta^{m}\omega) \right\} \right ) \\ 
        & \le \frac{1}{k} \sum_{m=0}^{k\lfloor\frac{n}{k} \rfloor -1} \log\left ( \max\left\{t, c_{k}^{(\beta)}(\theta^{m}\omega)  \right\}\right ) 
        + \frac{1}{tk} \sum_{m=0}^{k\lfloor\frac{n}{k} \rfloor -1} \epsilon^{(k)}(\theta^{m}\omega,\delta),
    \end{split}
    \end{equation}
    provided that $|\beta-\alpha| < \delta$. If $\beta$ is chosen to be rational, then~\eqref{eq:birkhoff-app-1}--\eqref{eq:birkhoff-app-2} imply
    \begin{equation}\label{eq:limsupresult}
    \begin{split}
    \limsup_{n\to \infty} \frac{1}{n}  \log c_{n}^{(\alpha)}(\omega) 
        &\leq \frac{1}{k}  \int \log \left (\max \left\{ t, c_{k}^{(\beta)} \right\}\right ) \dd\pp + \frac{1}{tk} \int \epsilon^{(k)}(\,\cdot\,,\delta) \dd\pp \\
        &\leq \frac{1}{k}  \int \log \left (\max \left\{ t, c_{k}^{(\alpha)}\right\} \right ) \dd\pp + \frac{2}{tk} \int \epsilon^{(k)}(\,\cdot\,,\delta) \dd\pp ,  
        \end{split}
    \end{equation}
    where we have used \eqref{eq:appendixequicon} once more in the last step. Taking $\delta \to 0$ using~\eqref{eq:mod-int-to-0}, we deduce that
    \begin{equation}
        \limsup_{n\to \infty} \frac{1}{n}  \log c_{n}^{(\alpha)}(\omega) \leq \frac{1}{k}  \int \log \left (\max\left\{t, c_{k}^{(\alpha)}\right\} \right ) \dd\pp .
    \end{equation}
    Taking $t \to 0$ and using Lebesgue dominated convergence, we deduce that
    \begin{equation}
        \limsup_{n\to \infty} \frac{1}{n}  \log c_{n}^{(\alpha)}(\omega) \leq \frac{1}{k}  \int \log \left (c_{k}^{(\alpha)} \right ) \dd\pp.
    \end{equation}
    Finally, since $k$ was arbitrary, this yields~\eqref{eq:appendixlemma}.
\end{proof}

\subsection{Vanishing of ergodic differences}
\begin{lemma}
\label{lem:shift-cancellation}
    If a measurable function $g : \Omega \to \rr$ is such that
    $g \circ \theta - g \in L^1 (\Omega,\dd\pp)$, then
    \[ 
        \int (g\circ \theta - g) \dd\pp = 0.
    \]
\end{lemma}

\begin{remark}If $g\in L^1$, then the result follows immediately from linearity of the integral: $\int(g\circ \theta -g) \dd\pp= \int g\circ\theta \dd\pp -\int g\dd\pp=0$, since $\theta$ is measure preserving. The point of the lemma is that the identity continues to hold for non-integrable $g$.
\end{remark}

\begin{proof}
    Let $I$ denote the value of the integral of interest and suppose for the sake of contradiction that $I \neq 0$. By Birkhoff's theorem, the following holds true for $\pp$-almost all~$\omega$:
    \begin{align*}
     I &= \lim_{N\to\infty} \frac 1N \sum_{n=0}^{N-1} g(\theta^{n+1}\omega) - g(\theta^{n}\omega) \\
            &= \lim_{N\to\infty} \frac 1N \left[ g(\theta^{N}\omega) - g(\omega)\right] \\
            &= \lim_{N\to\infty} \frac 1N g(\theta^{N}\omega).
    \end{align*}
    This almost sure convergence implies convergence in probability, so
    \[
        \lim_{N\to\infty} \pp \left\{\omega \ : \ \left| \frac 1N g(\theta^{N}\omega) -  I \right| \leq \epsilon \right\} = 1
    \] 
    for every $\epsilon > 0$. In particular, we should be able to choose $2\epsilon = |I|$ since we are assuming that $I \neq 0$.
    However, using basic manipulations and the fact that $\pp$ is $\theta$ invariant, 
    \begin{align*}
        \pp \left\{\omega \ : \ \left| \frac 1N g(\theta^{N}\omega) -  I \right| \leq \left|\frac{I}{2}\right| \right\}
            &\leq \pp \left\{\omega \ : \  \left| g(\theta^N\omega) \right| \geq  N \left|\frac{I}{2}\right| \right \} \\
            &=  \pp \left\{\omega \ : \  \left|g(\omega)\right| \geq  N \left|\frac{I}{2}\right| \right \}
    \end{align*}
    must decay to $0$ as $N\to\infty$ because $g$ is real valued, $I \neq 0$ and $\pp$ is a probability measure. This is the desired contradiction, forcing us to conclude that $I = 0$.
\end{proof}

\newcommand{\etalchar}[1]{$^{#1}$}

\end{document}